\PassOptionsToPackage{ocgcolorlinks,linkcolor=BlueViolet,citecolor=OliveGreen,urlcolor=RawSienna,hyperfootnotes=false}{hyperref} 

\documentclass[11pt, a4paper]{article}

\pdfoutput=1

\usepackage[british]{babel}
\addto\captionsbritish{}

\usepackage{mathtools} 
\numberwithin{equation}{section} 

\usepackage{amssymb}
\usepackage{bbm} 

\usepackage[a4paper,top=2in,bottom=1.5in,left=1in,right=1in]{geometry}

\usepackage{enumitem} 

\makeatletter
\renewcommand\subparagraph{\@startsection{subparagraph}{5}%
	{\parindent}
	{0pt}
	{-1em}
	{\normalfont\itshape}
} 
\makeatother

\usepackage[dvipsnames]{xcolor}

\usepackage{graphicx} 

\usepackage[bf,small,skip=0pt]{caption} 
\usepackage{tikz}
\usetikzlibrary{calc} 
\usetikzlibrary{intersections} 
\usetikzlibrary{patterns} 

\usepackage{amsthm} 

\newtheorem{theorem}{Theorem}
\newtheorem{lemma}{Lemma} 
\newtheorem{proposition}{Proposition}
\newtheorem{corollary}{Corollary}

\theoremstyle{definition}
\newtheorem{remark}{Remark}

\usepackage[numbers,sort&compress,merge]{natbib}
\let\originalleft\left
\let\originalright\right
\renewcommand{\left}{\mathopen{}\mathclose\bgroup\originalleft}
\renewcommand{\right}{\aftergroup\egroup\originalright}

\def\[{\begin{equation}}
\def\]{\end{equation}}

\newcommand{\D}{\mathrm{d}}
\DeclareMathOperator{\E}{e}
\newcommand{\I}{\mathrm{i}}

\DeclarePairedDelimiter{\bra}{\langle}{\rvert}
\DeclarePairedDelimiter{\ket}{\lvert}{\rangle}

\DeclareMathOperator*{\ordprodopp}{\prod\limits^{\vbox to -.5ex{\kern-0.5ex\hbox{$\leftharpoonup$}\vss}}}
\DeclareMathOperator*{\ordprod}{\prod\limits^{\vbox to -.5ex{\kern-0.5ex\hbox{$\rightharpoonup$}\vss}}}

\newcommand{\Vector}[1]{\boldsymbol{#1}}

\newcommand{\To}{\cdots\mspace{-1mu}}

\DeclareMathOperator{\End}{End}
\DeclareMathOperator{\id}{\mathbbm{1}} 

 \let\Im\undefined \DeclareMathOperator{\Im}{Im}
  
\DeclareMathOperator{\Res}{Res}

\title{\textbf{Integral formula for elliptic SOS models \\ with domain walls and a reflecting end}}

\author{\normalsize \textsc{Jules Lamers} \\[\baselineskip] \normalsize \textit{Institute for Theoretical Physics} \\ \normalsize \textit{Center for Extreme Matter and Emergent Phenomena, Utrecht University} \\ \normalsize \textit{Leuvenlaan 4, 3584 CE Utrecht, The Netherlands} \\[\baselineskip] \normalsize \texttt{j.lamers@uu.nl}}

\date{}

\begin{document}

\begin{flushright}
	\footnotesize
	ITP-UU-15/11
\end{flushright}

{\let\newpage\relax\maketitle} 
\thispagestyle{empty}


\begin{abstract}
\noindent In this paper we extend previous work of Galleas and the author to elliptic \textsc{sos} models. We demonstrate that the dynamical reflection algebra can be exploited to obtain a functional equation characterizing the partition function of an elliptic \textsc{sos} model with domain-wall boundaries and one reflecting end. Special attention is paid to the structure of the functional equation. Through this approach we find a novel multiple-integral formula for that partition function.
\end{abstract}

\begin{center}
\small{Keywords: elliptic \textsc{sos} models, functional equations, open boundaries, domain walls \\ \textsc{pacs} numbers: 05.50.+q, 02.30.Ik} 
\bigskip

\date{October 2015}
\end{center}

\tableofcontents

\section{Introduction}
\label{sec:intro}

\emph{Solid-on-solid} (\textsc{sos}) models are Ising-type models for the growth of crystal surfaces. The `solid-on-solid condition' \cite{Ben_84} forbids the surface to have overhangs, so that the shape of the surface can be modelled by associating discrete height variables to the sites of a square lattice. The surface tension is taken into account by interactions between these height variables. In this work we are only interested in the exactly-solvable subclass consisting of models that are of `interaction-round-a-face' type, i.e.\ the interactions take place between the four vertices sharing a common face, and for which the height difference between all adjacent vertices is precisely one. Accordingly we take `\textsc{sos} models' to mean exclusively these models, as is standard in the literature on quantum integrability. The matrix containing the face weights satisfies a `face' (quantum) Yang-Baxter equation (\textsc{ybe}) guaranteeing the commutativity of the corresponding transfer matrices. The partition function is the product of the local weights, summed over all height configurations for the lattice.

Alternatively, \textsc{sos} models can be interpreted as `generalized six-vertex models'. The \textsc{sos} weight-matrix corresponds to a `generalized' or \emph{dynamical $R$-matrix} that depends on a \emph{dynamical} parameter keeping track of the original heights. Importantly, the dynamical $R$-matrix satisfies the ice rule. In addition it obeys the \emph{dynamical Yang-Baxter equation}, which was first encountered by Gervais and Neveu \cite{GN_84} in the context of (Liouville) conformal field theory and later independently obtained by Felder \cite{Fel_94b} as the quantization of the modified classical \textsc{ybe}. The appropriate quantum-algebraic setting is that of elliptic quantum groups \cite{Fel_94b, EF_98}. The ordinary six-vertex model may be recovered in the limit of infinite heights.

In the context of integrability, \textsc{sos} models were introduced by Baxter~\cite{Bax_73} as a convenient reformulation of the symmetric eight-vertex model. In algebraic terms~\cite{TF_79, FV_96b} the \emph{vertex-face} transformation relating the two models is a local (gauge) transformation mapping the eight-vertex $R$-matrix to the dynamical $R$-matrix. By the ice rule the latter allows for the construction of a highest-weight vector (pseudovacuum), opening the way for a Bethe-ansatz analysis. Like the eight-vertex $R$-matrix, the associated dynamical $R$-matrix involves elliptic functions; this elliptic \textsc{sos} model is sometimes referred to as the `eight-vertex \textsc{sos}' (\textsc{8vsos}) model. The dynamical $R$-matrix also appears in the quantization of the elliptic Ruijsenaars-Schneider model~\cite{ACF_97}, explaining the adjective `dynamical': here the dynamical parameter is a coordinate on the phase space of the classical integrable model.

Besides the relation with various models in statistical physics, \textsc{sos} models have several applications in physics and mathematics, depending on the choice of boundary conditions. The case with fixed \emph{domain-wall} boundary conditions, originating in the study of scalar products of Bethe vectors \cite{Kor_82}, is of particular interest for enumerative combinatorics in the spirit of Kuperberg~\cite{Kup_95}, whose work was generalized to elliptic \textsc{sos} models at root of unity by Rosengren~\cite{Ros_09}. Models with domain-wall boundaries further allow for concise expressions for the partition function as certain determinants. For the six-vertex model the partition function takes the form of a single determinant~\cite{Ize_87}, while for \textsc{sos} models it can be written as a sum of $2^L$ determinants for models on an $L\times L$ lattice~\cite{Ros_09}. Interestingly, the presence of one \emph{reflecting} (open) boundary makes it possible get a single determinant for six-vertex models \cite{Tsu_98} as well as for \textsc{sos} models \cite{FK_10, Fil_11}. Such determinant-formulas are elegant and helped uncovering striking physics like the possible dependence of thermodynamic quantities on the choice of boundary conditions imposed at finite system size in the computation~\cite{KZ_00, *TRK_15a, *TRK_15b}. Yet many applications of partition functions reside in the homogeneous limit, whose evaluation can be rather nontrivial~\cite{ICK_92}, especially when elliptic functions are involved. This is one issue that makes it interesting to find alternative ways to express partition function. For \textsc{sos} models with domain walls at all boundaries such alternatives are available~\cite{PRS_08, Gal_12, Gal_13a}, and also for the six-vertex model with one reflecting end~\cite{GL_14}. In \cite{Gal_12, Gal_13a, GL_14} multiple-integral formulas were obtained as the solution of certain functional equations. Such repeated contour integrals might allow one to access the thermodynamic limit more straightforwardly. In the present work we extend \cite{GL_14} to the case of elliptic \textsc{sos} models with domain-wall boundaries and one reflecting end, complementing the result of \cite{Fil_11}.

Functional equations are at the core of exactly solvable models of statistical mechanics. The functional equations derived in \cite{Gal_12, Gal_13a, GL_14} have particular origins: the dynamical Yang-Baxter algebra for \cite{Gal_12, Gal_13a} and the reflection algebra for \cite{GL_14}. The use of such algebras as a source of functional equations was put forward by Galleas in \cite{Gal_08} for spectral problems and in \cite{Gal_10} for vertex-model partition functions. Compact ways to write solutions of such functional equations were only obtained after several refinements of the method~\cite{Gal_11, Gal_12, Gal_13a, Gal_14, *Gal_13b, GL_14}. At the same time this \emph{algebraic-functional method} is a beautiful example of the rigid structure imposed by the (quantum) algebra underlying exactly solvable models in statistical physics, which is reflected in the many nice properties of the functional equations obtained in this way.

\paragraph{Outline.} This paper is organized as follows. To fix our notation and conventions we describe the elliptic \textsc{sos} model with domain-wall boundary conditions and one reflecting end in Section~\ref{sec:setup}, setting up a graphical notation for generalized six-vertex models along the way. Our main results are stated at the start of Section~\ref{sec:results}. The first result is a linear functional equation for the model's partition function, which we derive using the algebraic-functional method in Section \ref{sec:AF}. Secondly we show that, despite being linear, this functional equation has a \emph{unique} analytic solution, up to a constant overall factor. This is a consequence of a recursive relation between the functional equations for successive system sizes, established in Section~\ref{sec:reduction}. Our third result is a new multiple-integral formula for the partition function, which is found by solving our functional equation by recursion in Section~\ref{sec:sol}. We conclude with some remarks and open questions in Section~\ref{sec:conclusion}. Several technical details are left for the appendices.

\section{Set-up}
\label{sec:setup}

Consider an \textsc{sos} model with height variables taking values in $\theta + \gamma\mathbb{Z}$. The \emph{dynamical} parameter~$\theta$ sets a reference height, whilst the \emph{crossing} or \emph{anisotropy} parameter~$\gamma$ controls the step size. The condition that neighbouring heights differ by one allows for six types of height profiles around any face, each with an associated local (Boltzmann) weight as shown in Figure~\ref{fig:SOS_weights}. 

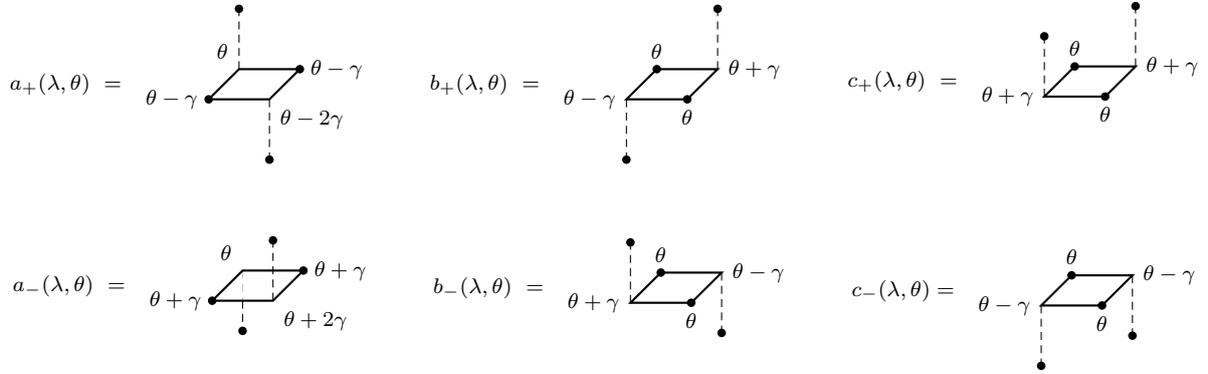
\begin{figure}[h]
	\begin{center}
		\begin{tikzpicture}[baseline={([yshift=-.5*11pt*0.8]current bounding box.center)},font=\scriptsize,scale=0.8,thick,x={(1cm,0)},y={(.5cm,0.5cm)}, z={(0cm,1cm)}]
		\node[left] at (-1.5,.5,0) {$a_+(\lambda,\theta) \ = $};
		
		\draw[thin,densely dashed] (0,1,0) -- (0,1,1);
		\draw[thin,densely dashed] (1,0,0) -- (1,0,-1);
		
		\draw (0,0,0) node[left]{$\theta-\gamma$} -- (1,0,0) node[below right]{$\theta-2\gamma$} -- (1,1,0) node[right]{$\theta-\gamma$} -- (0,1,0) node[above left]{$\theta$} -- cycle;
		
		\fill [black] (0,0,0) circle (2pt);
		\fill [black] (1,0,-1) circle (2pt); 
		\fill [black] (1,1,0) circle (2pt);
		\fill [black] (0,1,1) circle (2pt);
		\end{tikzpicture}
		\hskip .5cm
		\begin{tikzpicture}[baseline={([yshift=-.5*11pt*0.8]current bounding box.center)},font=\scriptsize,scale=0.8,thick,x={(1cm,0)},y={(.5cm,0.5cm)}, z={(0cm,1cm)}]
		\node[left] at (-1.5,.5,0) {$b_+(\lambda,\theta) \ = $};
		
		\draw[thin,densely dashed] (0,0,0) -- (0,0,-1);
		\draw[thin,densely dashed] (1,1,0) -- (1,1,1);
		
		\draw (0,0,0) node[left]{$\theta-\gamma$} -- (1,0,0) node[below]{$\theta$} -- (1,1,0) node[right]{$\theta+\gamma$} -- (0,1,0) node[above]{$\theta$} -- cycle;
		
		\fill [black] (0,0,-1) circle (2pt);
		\fill [black] (1,0,0) circle (2pt); 
		\fill [black] (1,1,1) circle (2pt);
		\fill [black] (0,1,0) circle (2pt); 
		\end{tikzpicture}
		\hskip .5cm
		\begin{tikzpicture}[baseline={([yshift=-.5*11pt*0.8]current bounding box.center)},font=\scriptsize,scale=0.8,thick,x={(1cm,0)},y={(.5cm,0.5cm)}, z={(0cm,1cm)}]
		\node[left] at (-1.5,.5,0) {$c_+(\lambda,\theta) \ = $};
		
		\draw[thin,densely dashed] (0,0,0) -- (0,0,1);
		\draw[thin,densely dashed] (1,1,0) -- (1,1,1);
		
		\draw (0,0,0) node[left]{$\theta+\gamma$} -- (1,0,0) node[below]{$\theta$} -- (1,1,0) node[right]{$\theta+\gamma$} -- (0,1,0) node[above]{$\theta$} -- cycle;
		
		\fill [black] (0,0,1) circle (2pt);
		\fill [black] (1,0,0) circle (2pt); 
		\fill [black] (1,1,1) circle (2pt);
		\fill [black] (0,1,0) circle (2pt); 
		
		\node at (0,0,-1) {}; 
		\end{tikzpicture}
		\\[\baselineskip]
		\begin{tikzpicture}[baseline={([yshift=-.5*11pt*0.8]current bounding box.center)},font=\scriptsize,scale=0.8,thick,x={(1cm,0)},y={(.5cm,0.5cm)}, z={(0cm,1cm)}]
		\node[left] at (-1.5,.5,0) {$a_-(\lambda,\theta) \ = $};
		
		\draw[gray!50,thin,densely dashed] (0,1,0) -- (0,1,-.5);
		\draw[thin,densely dashed] (0,1,-.5) -- (0,1,-1);
		\draw[thin,densely dashed] (1,0,0) -- (1,0,1);
		
		\draw (0,0,0) node[left]{$\theta+\gamma$} -- (1,0,0) node[below right]{$\theta+2\gamma$} -- (1,1,0) node[right]{$\theta+\gamma$} -- (0,1,0) node[above left]{$\theta$} -- cycle;
		
		\fill [black] (0,0,0) circle (2pt);
		\fill [black] (1,0,1) circle (2pt); 
		\fill [black] (1,1,0) circle (2pt);
		\fill [black] (0,1,-1) circle (2pt);
		\end{tikzpicture}
		\hskip .5cm
		\begin{tikzpicture}[baseline={([yshift=-.5*11pt*0.8]current bounding box.center)},font=\scriptsize,scale=0.8,thick,x={(1cm,0)},y={(.5cm,0.5cm)}, z={(0cm,1cm)}]
		\node[left] at (-1.5,.5,0) {$b_-(\lambda,\theta) \ = $};
		
		\draw[thin,densely dashed] (0,0,0) -- (0,0,1);
		\draw[thin,densely dashed] (1,1,0) -- (1,1,-1);
		
		\draw (0,0,0) node[left]{$\theta+\gamma$} -- (1,0,0) node[below]{$\theta$} -- (1,1,0) node[right]{$\theta-\gamma$} -- (0,1,0) node[above]{$\theta$} -- cycle;
		
		\fill [black] (0,0,1) circle (2pt);
		\fill [black] (1,0,0) circle (2pt); 
		\fill [black] (1,1,-1) circle (2pt);
		\fill [black] (0,1,0) circle (2pt); 
		\end{tikzpicture}
		\hskip .5cm
		\begin{tikzpicture}[baseline={([yshift=-.5*11pt*0.8]current bounding box.center)},font=\scriptsize,scale=0.8,thick,x={(1cm,0)},y={(.5cm,0.5cm)}, z={(0cm,1cm)}]
		\node[left] at (-1.5,.5,0) {$c_-(\lambda,\theta) = $};
		
		\draw[thin,densely dashed] (0,0,0) -- (0,0,-1);
		\draw[thin,densely dashed] (1,1,0) -- (1,1,-1);
		
		\draw (0,0,0) node[left]{$\theta-\gamma$} -- (1,0,0) node[below]{$\theta$} -- (1,1,0) node[right]{$\theta-\gamma$} -- (0,1,0) node[above]{$\theta$} -- cycle;
		
		\fill [black] (0,0,-1) circle (2pt);
		\fill [black] (1,0,0) circle (2pt); 
		\fill [black] (1,1,-1) circle (2pt);
		\fill [black] (0,1,0) circle (2pt); 
		
		\node at (0,1,1) {}; 
		\end{tikzpicture}
	\end{center}
	\caption{In \textsc{sos} models where all neighbouring heights differ by one unit there are six allowed height profiles around a face.}
	\label{fig:SOS_weights}
\end{figure}

We find it convenient to switch to the language of generalized vertex models. This is done via the following procedure, which seems to have been first pointed out by Lenard \cite[Note added in proof]{Lie_67}; see also \cite{Bei_77}. Pass to the dual lattice, for which the height variables are associated to the dual faces.Assign classical spins to the dual edges according to the following rule: going around a dual vertex in anti-clockwise fashion, place an arrow that points outwards (inwards) if the height increases (decreases) by one unit. Upon completing a circle one returns to the height one started at, which means that the arrows around the dual vertex must satisfy the ice rule from the six-vertex model. In the presence of a reflecting boundary it is useful to convert the arrows into signs, where we take a plus for arrows following some fixed orientation of the line (ordinarily: going up or to the right), see Figure~\ref{fig:vertex_weights}.
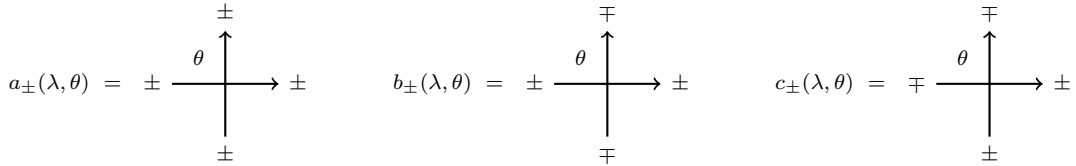
\begin{figure}[h]
	\begin{center}
		\begin{tikzpicture}[baseline={([yshift=-.5*11pt*0.7]current bounding box.center)},font=\scriptsize,scale=0.7,thick]
		\node[left] at (-.75,1) {$a_\pm(\lambda,\theta) \ = $};
		\draw[->] (0,1) node[left]{$\pm$} -- (2,1) node[right]{$\pm$};
		\draw[->] (1,0) node[below]{$\pm$} -- (1,2) node[above]{$\pm$};
		\node at (.5,1.5) {$\theta$};
		\end{tikzpicture}
		\hskip .75cm
		\begin{tikzpicture}[baseline={([yshift=-.5*11pt*0.7]current bounding box.center)},font=\scriptsize,scale=0.7,thick]
		\node[left] at (-.75,1) {$b_\pm(\lambda,\theta) \ = $};
		\draw[->] (0,1) node[left]{$\pm$} -- (2,1) node[right]{$\pm$};
		\draw[->] (1,0) node[below]{$\mp$} -- (1,2) node[above]{$\mp$};
		\node at (.5,1.5) {$\theta$};
		\end{tikzpicture}
		\hskip .75cm
		\begin{tikzpicture}[baseline={([yshift=-.5*11pt*0.7]current bounding box.center)},font=\scriptsize,scale=0.7,thick]
		\node[left] at (-.75,1) {$c_\pm(\lambda,\theta) \ = $};
		\draw[->] (0,1) node[left]{$\mp$} -- (2,1) node[right]{$\pm$};
		\draw[->] (1,0) node[below]{$\pm$} -- (1,2) node[above]{$\mp$};
		\node at (.5,1.5) {$\theta$};
		\end{tikzpicture}
	\end{center}
	\caption{The six vertex weights dual to the face weights from Figure~\ref{fig:SOS_weights}.}
	\label{fig:vertex_weights}
\end{figure}

The central quantity in this work is the partition function $\mathcal{Z}$ of an elliptic \textsc{sos} model on a rectangular lattice with $2L$ rows and $L$ columns, and specific boundary conditions: one reflecting end and domain walls at the three other ends, as depicted in Figure~\ref{fig:part_fn}. In this section we set up our notation and conventions, which mostly follow \cite{Gal_12, Gal_13a, GL_14}, culminating in an algebraic expression for $\mathcal{Z}$. The set-up is a generalization of that in \cite{GL_14} including (shifts of) the dynamical parameter~$\theta$ in the appropriate places.

\begin{figure}[h]
\begin{center}
\begin{tikzpicture}[baseline={([yshift=-.5*11pt*0.8]current bounding box.center)},font=\scriptsize,scale=0.8,thick]
	\foreach \x in {0,...,3} \draw (\x+.5,.5) -- +(0,4.5);
	\foreach \y in {0,...,6} \draw (.5,.75*\y+.5) -- +(3,0);
	
	\foreach \y in {1,...,3} \draw (.5,1.5*\y-.25) -- ++(-.5,.75) -- ++(0,-1.5) -- +(.5,.75);
	\foreach \y in {0,...,3} {
		\draw[thin,densely dashed] (.5,1.5*\y+.5) -- +(-.5,0);
		\fill [black] (0,1.5*\y+.5) circle (2pt);
	}
	\foreach \x in {0,...,3} \foreach \y in {0,...,6} \fill [black] (\x+.5,.75*\y+.5) circle (2pt);
	
	\node at (0,.125) {$\theta$}; 
	\node at (.5,.125) {$\theta$}; 
	\node at (1.5,.125) {$\theta-\gamma$};
	\node at (2.75,.125) {$\cdots$};
	\node[right] at (3.5,.125) {$\theta-L\gamma$}; 
	\node[right] at (3.5,1.25) {$\theta-(L-1)\gamma$}; 
	\node[right,yshift=3pt] at (3.75,2){$\vdots$};
	\node[right,yshift=3pt] at (3.75,3.5){$\vdots$};
	\node[right] at (3.5,4.25) {$\theta+(L-1)\gamma$}; 
	\node[right] at (3.5,5.375) {$\theta+L\gamma$}; 
	\node at (2.75,5.375) {$\cdots$};
	\node at (1.5,5.375) {$\theta+\gamma$};
	\node at (.5,5.375) {$\theta$};
	\node at (0,5.375) {$\theta$};
\end{tikzpicture}
\hskip 3cm
\begin{tikzpicture}[baseline={([yshift=-.5*11pt*0.8]current bounding box.center)},font=\scriptsize,scale=0.8,thick]
	\foreach \y in {1,...,3} {
		\draw[rounded corners=5pt]    (0,1.5*\y-.25) -- ++(.7,-.35) -- +(3.3,0) node[right]{$-$};
		\draw[rounded corners=5pt,->] (0,1.5*\y-.25) -- ++(.7,+.35) -- +(3.3,0) node[right]{$+$};
	}
	\fill[gray!30] (0,0) rectangle (-.15,5.5) ; \draw (0,0) -- (0,5.5);
	
	\foreach \x in {1,...,3} \draw[->] (\x,0) node[below]{$+$} -- +(0,5.5) node[above]{$-$};
	
	\node at (.5,.5) {$\theta$}; 
\end{tikzpicture}
\end{center}
\caption{The lattice of an \textsc{sos} model with domain-wall boundaries and one reflecting end (left) and the dual lattice with the corresponding generalized six-vertex model (right).}
\label{fig:part_fn}
\end{figure}
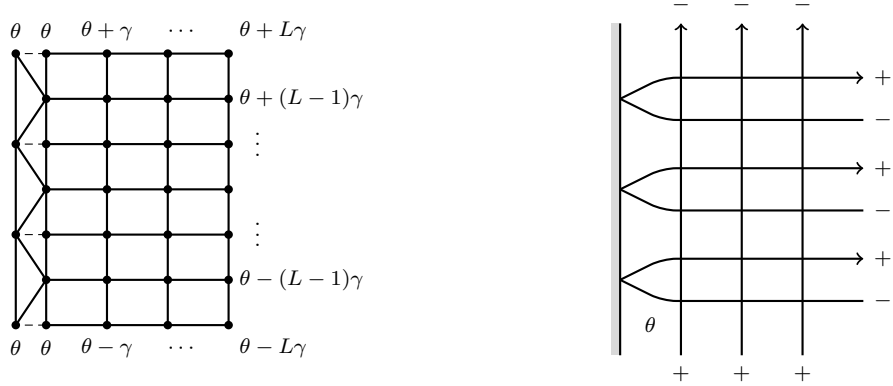

\subsection{Bulk: dynamical Yang-Baxter algebra}

Let $V \coloneqq \mathbb{C} \, e_+ \oplus \mathbb{C} \, e_-$ be a complex two-dimensional vector space with standard basis vectors~$e_\pm$, to which we assign \emph{weights} $\mathrm{wt}(e_\pm) = \pm 1$. This grading of $V$ by weights is conveniently accounted for via the $\mathfrak{sl}_2$ Cartan generator on $V$, which we denote by $h \in\End(V)$ and acts as $h \, e_\pm = \mathrm{wt}(e_\pm) \, e_\pm$. Write $\mathfrak{h} \subseteq \mathfrak{sl}_2$ for the corresponding Cartan subalgebra. 

Let $\lambda,\theta,\gamma\in\mathbb{C}$ be generic complex numbers. We think of the spectral parameter~$\lambda$ and the dynamical parameter~$\theta$ as variables, and of the crossing parameter~$\gamma$ as fixed.

\paragraph{Dynamical \textit{R}-matrix.} The statistical weights of the height configurations in an \textsc{sos} model are encoded in the dynamical \textit{R}-matrix $R(\lambda,\theta) \in \End(V\otimes V)$,
which must be invertible for generic $\lambda$ and~$\theta$ and is required to satisfy the \emph{dynamical Yang-Baxter equation} on $V\otimes V\otimes V$:
	\[ \label{dybe} 
	\begin{aligned} 
	& R_{12}(\lambda_1-\lambda_2,\theta-\gamma \, h_3) \ R_{13}(\lambda_1-\lambda_3,\theta) \ R_{23}(\lambda_2-\lambda_3,\theta-\gamma \,h_1) \\ 
	& \qquad\qquad = \ R_{23}(\lambda_2-\lambda_3,\theta) \ R_{13}(\lambda_1-\lambda_3,\theta -\gamma\, h_2) \ R_{12}(\lambda_1-\lambda_2,\theta) \ . 
	\end{aligned} 
	\]
Here we employ the usual tensor-leg notation, where subscripts indicate on which copies of $V$ in the tensor product the operators act nontrivially; for example,
	\[ \label{R12;3}
	R_{12}(\lambda,\theta-\gamma \,h_3) \, e_1 \otimes e_2 \otimes e_\pm = \big( R(\lambda,\theta \mp \gamma) \, e_1 \otimes e_2 \big) \otimes e_\pm \ ,
	\]
or equivalently, as a matrix acting on the third copy in $V\otimes V\otimes V$,
	\[
	R_{12}(\lambda,\theta-\gamma \,h_3) = \begin{pmatrix} R_{12}(\lambda,\theta-\gamma) & 0 \\ 0 & R_{12}(\lambda,\theta+\gamma) \end{pmatrix}_{\!3} \ .
	\]
Through the dynamical parameter each \textit{R}-matrix in \eqref{dybe} is sensitive to the weights in any copy of $V$ present on the left of the $V\otimes V$ on which that \textit{R}-matrix acts. This can also be seen in terms of a graphical notation in which one depicts the dynamical \textit{R}-matrix and its entries, labelled by indices $\alpha,\beta,\alpha',\beta'\in\{\pm\}$, by including $\theta$ in the usual notation:
\begin{equation}
	R_{12}(\lambda_1-\lambda_2,\theta) \ = \
	\begin{tikzpicture}[baseline={([yshift=-.5*11pt*0.8+8pt]current bounding box.center)}, scale=0.8,font=\scriptsize,thick]
		\draw[->] (0,1) node[left]{$\lambda_1$} -- (2,1);
		\draw[->] (1,0) node[below]{$\lambda_2$} -- (1,2);
		
		\node at (.5,1.5) {$\theta$};
	\end{tikzpicture} 
	\ \ , \qquad
	R_{\alpha,\,\beta}^{\alpha'\!,\,\beta'}(\lambda,\theta) \ = \
	\begin{tikzpicture}[baseline={([yshift=-.5*11pt*0.8]current bounding box.center)}, scale=0.8,font=\scriptsize,thick]
		\draw[->] (0,1) node[left]{$\alpha$} -- (2,1) node[right]{$\alpha'$};
		\draw[->] (1,0) node[below]{$\beta$} -- (1,2) node[above]{$\beta'$};
		\draw[thin,densely dotted,->] (1-.4,1) arc(180:270:.4);
		
		\node at (.5,1.5) {$\theta$};
		\node at (.6,.6) {$\lambda$};
	\end{tikzpicture}
	\ .
\end{equation}
The arrows serve to keep track of the orientation as we will need to rotate this picture. The dynamical Yang-Baxter equation~\eqref{dybe} now acquires the usual graphical form as for the six-vertex model, but with fixed height~$\theta$ in the left-most face of both sides of the equation:
\begin{equation}
	\begin{tikzpicture}[baseline={([yshift=-.5*11pt*0.6]current bounding box.center)},scale=0.6,font=\scriptsize,thick]
	\pgfmathsetmacro{\csc}{1/sin(130)}
		\draw[->] (-130:1.5*\csc) node[below]{$\lambda_1$} -- (50:1.5*\csc);
		\draw[->] (.5,-1.5) node[below]{$\lambda_2$} -- (.5,1.5);
		\draw[->] (-50:1.5*\csc) node[below]{$\lambda_3$} -- (130:1.5*\csc);
		
		\node at (-1,0) {$\theta$};
	\end{tikzpicture}
	\ \ = \ \
	\begin{tikzpicture}[baseline={([yshift=-.5*11pt*0.6]current bounding box.center)},scale=0.6,font=\scriptsize,thick]
	\pgfmathsetmacro{\csc}{1/sin(130)}
		\draw[->] (-130:1.5*\csc) node[below]{$\lambda_1$} -- (50:1.5*\csc);
		\draw[->] (-.5,-1.5) node[below]{$\lambda_2$} -- (-.5,1.5); 
		\draw[->] (-50:1.5*\csc) node[below]{$\lambda_3$} -- (130:1.5*\csc);
		
		\node at (-1.25,0) {$\theta$};
	\end{tikzpicture}
	\ .
\end{equation}

Consider the subspace $\End_\mathfrak{h} (V\otimes V) \subseteq \End (V\otimes V)$ of $\mathfrak{h}$-invariant, i.e.\ weight-preserving, operators on $V \otimes V$. In other words, an element of $\End_\mathfrak{h} (V\otimes V)$ satisfies the ice rule: it commutes with $h_1 + h_2 \in V \otimes V$. The \textsc{sos} model that we are interested in corresponds to the elliptic solution of \eqref{dybe} lying in $\End_\mathfrak{h} (V\otimes V)$,
	\[ \label{rmat}
	R(\lambda,\theta) = \begin{pmatrix} 
	a_+(\lambda,\theta) & 0 & 0 & 0 \\ 
	0 & b_+(\lambda,\theta) & c_+(\lambda,\theta) & 0 \\ 
	0 & c_-(\lambda,\theta) & b_-(\lambda,\theta) & 0 \\ 
	0 & 0 & 0 & a_-(\lambda,\theta) 
	\end{pmatrix} \, .
	\]
The entries of this dynamical $R$-matrix, cf.\ Figure~\ref{fig:SOS_weights}, are given by
	\[ \label{weights} 
	a_\pm(\lambda,\theta) = f(\lambda+\gamma) \ , \quad b_\pm(\lambda,\theta) = f(\lambda) \, \frac{f(\theta\mp\gamma)}{f(\theta)} \ , \quad  c_\pm(\lambda,\theta) = \frac{f(\theta\mp\lambda)}{f(\theta)} \, f(\gamma) \ ,
	\]
where in turn $f(\lambda) \coloneqq -\I \, \E^{-\I \pi \tau/4} \, \vartheta_1(\I\lambda|\tau)/2$ is essentially the odd Jacobi theta function with elliptic nome $\E^{\I \pi \tau}\in\mathbb{C}$ satisfying $\Im(\tau) > 0$. The properties of $f$ required for this work are collected in Appendix~\ref{sec:theta}. The algebraic structure underlying \eqref{rmat}--\eqref{weights} is the elliptic quantum group~$E_{\tau, \gamma}(\mathfrak{sl}_2)$ whose properties are discussed in \cite{Fel_94b, FV_96a, EF_98}. To make contact with the ordinary six-vertex model one takes the trigonometric limit, $\lim_{\tau\to\I\infty} f(\lambda) = \sinh(\lambda)$, and subsequently lets $\theta\to\infty$, which the yields the [$U_q(\widehat{\mathfrak{sl}_2})$-invariant form of the] six-vertex model's \textit{R}-matrix considered in \cite{GL_14,Gal_12}.

We will also need the following \emph{unitarity} property of the \textit{R}-matrix:
	\[ \label{unitarity}
	R_{12}(\lambda,\theta) \ R_{21}(-\lambda,\theta) = f(\gamma+\lambda) \, f(\gamma-\lambda) \, \id \ .  
	\]
Here $R_{21}(\lambda,\theta) = P \, R_{12}(\lambda,\theta) \, P \in \End_\mathfrak{h} (V \otimes V)$ with $P \in\End_\mathfrak{h} (V \otimes V)$ the permutation operator. Notice that the proportionality factor in \eqref{unitarity} does not depend on the dynamical parameter. Finally we record that transposing the first leg of the $R$-matrix yields an operator satisfying the ice rule  $[-h_1+h_2,R_{12}^{t_1}(\lambda,\theta)]=0$ and the \emph{crossing symmetry}
	\[ \label{crossing}
	{- \sigma^y_1} \ {\colon\! R_{12}^{t_1}({-\lambda}-\gamma,\theta+\gamma \, h_2)\colon} \sigma^y_1 \ \frac{f(\theta-\gamma\,h_1)}{f(\theta)} = R_{21}(\lambda,\theta) \ .
	\]
Here $\sigma^y \in \End(V)$ is the second Pauli matrix, and the colons indicates that the $h_1$ is taken to act \emph{after} (on the left of) the transposed $R$-matrix.

\paragraph{Monodromy matrices.} A central role in the quantum inverse-scattering method (\textsc{qism})~\cite{TF_79} is played by the monodromy matrix. We are interested in its dynamical version, introduced in \cite{FV_96b}, which is built out of solutions of \eqref{dybe} lying in $\End_\mathfrak{h}(V \otimes V)$.

Focus on any single row of the lattice, with associated `auxiliary space' $V_0 \cong V$ and spectral parameter $\lambda\in\mathbb{C}$. Every column $1 \leq j\leq L$ comes with a `local quantum space' $V_j \cong V$. The `global quantum space' is defined as the $L$-fold tensor product $W \coloneqq V_1 \otimes V_2 \otimes \cdots \otimes V_L$. Since \eqref{dybe} depends only on differences of spectral parameters we can easily incorporate a (fixed) inhomogeneity parameter~$\mu_j \in\mathbb{C}$ for each $V_j$, $1\leq j\leq L$.

There are two monodromy matrices $T_0(\lambda,\theta), \bar{T}_0(\lambda,\theta) \in \End (V_0 \otimes W)$,
defined as oppositely ordered products
	\[  \label{mono}
	\begin{aligned}
	& T_0(\lambda,\theta) \coloneqq \ordprod_{1 \leq j \leq L} R_{0 j} (\lambda - \mu_j ,\theta - \gamma \sum_{i=j+1}^L h_i) \ = 
	\begin{tikzpicture}[baseline={([yshift=-.5*11pt*.8+8pt]current bounding box.center)},scale=0.8,font=\scriptsize,thick]
		\draw[->] (0,1) node[left]{$\lambda$} -- (2.6,1) (3.4,1) -- (5,1);
		
		\draw[->] (1,0) node[below]{$\mu_L$} -- (1,2);
		\draw[->] (2,0) node[below]{$\mu_{L-1}$} -- (2,2);
		\draw[->] (4,0) node[below]{$\mu_1$} -- (4,2);
		
		\foreach \y in {-1,...,1} \draw (3+.2*\y,1) node{$\cdot\mathstrut$};
		
		\node at (.5,1.5) {$\theta$};
		\node at (1.5,1.5) {$\theta\mp\gamma$};
		\node at (2.5,1.5) {$\cdots$};
	\end{tikzpicture} 
	\ \ , \\
	& \bar{T}_0(\lambda,\theta) \coloneqq \ordprodopp_{L \geq j \geq 1} R_{j0} (\lambda + \mu_j ,\theta - \gamma \sum_{i=j+1}^L h_i) \ = \ 
	\begin{tikzpicture}[baseline={([yshift=-.5*11pt*.8+8pt]current bounding box.center)},scale=0.8,font=\scriptsize,thick]
		\draw[<-] (0,1) -- (2.6,1); \draw (3.4,1) -- (5,1) node [right]{${-\lambda}$};
		
		\draw[->] (1,0) node[below]{$\mu_L$} -- (1,2);
		\draw[->] (2,0) node[below]{$\mu_{L-1}$} -- (2,2);
		\draw[->] (4,0) node[below]{$\mu_1$} -- (4,2);
		
		\foreach \y in {-1,...,1} \draw (3+.2*\y,1) node{$\cdot\mathstrut$};
		
		\node at (.5,.5) {$\theta$};
		\node at (1.5,.5) {$\theta\mp\gamma$};
		\node at (2.5,.5) {$\cdots$};
	\end{tikzpicture}
	\ .
	\end{aligned}
	\]

The operator $H \coloneqq \sum_{j=1}^L h_j$, representing the generator of the Cartan subalgebra $\mathfrak{h} \subseteq \mathfrak{sl}_2$ on $W$, endows $W$ with a grading by weights: $H\, v = \mathrm{wt}(v) \, v$ for homogeneous vectors $v\in W$. Since $R(\lambda,\theta)\in \End_\mathfrak{h} (V \otimes V)$ we have $T_0(\lambda,\theta), \bar{T}_0(\lambda,\theta) \in \End_\mathfrak{h} (V_0 \otimes W)$, where the ice rule now entails commuting with $h_0+H \in V_0 \otimes W$.

Usually only one of the monodromy matrices in \eqref{mono} is considered since by \eqref{unitarity} the two are essentially inverse to each other:
	\[ \label{unitarity_mono} 
	T_0(\lambda,\theta) \ \bar{T}_0(-\lambda,\theta) = \prod_{j=1}^L f(\gamma-\lambda+\mu_j) f(\gamma+\lambda-\mu_j) \, \id_0 \otimes \id \ .
	\]
The two monodromy matrices are also related by crossing symmetry, following from the ice rule for the transposed $R$-matrix and \eqref{crossing},
	\[ \label{crossing_mono}
	(-1)^L \, \sigma^y_0 \ {\colon\! T_{0}^{t_0}({-\lambda}-\gamma,\theta+\gamma \, h_0)\colon} \sigma^y_0 \ \frac{f(\theta-\gamma\,H)}{f(\theta)} = \bar{T}_0(\lambda,\theta) \ .
	\]
For our purposes, however, it will be convenient to include both $T$ and $\bar{T}$.

\paragraph{Dynamical Yang-Baxter algebra.} The monodromy matrices~\eqref{mono} can be viewed as matrices in auxiliary space with entries acting on $W$:
	\[ \label{abcd}
	T_0(\lambda,\theta) = \begin{pmatrix}	
	A(\lambda,\theta) & B(\lambda,\theta) \\
	C(\lambda,\theta) & D(\lambda,\theta) 
	\end{pmatrix}_{\!0} \ , \qquad 
	\bar{T}_0(\lambda,\theta) = \begin{pmatrix}
	\bar{A}(\lambda,\theta) & \bar{B}(\lambda,\theta) \\
	\bar{C}(\lambda,\theta) & \bar{D}(\lambda,\theta) \end{pmatrix}_{\!0} \ .
	\]
Together with the ice rule, the dynamical Yang-Baxter equation~\eqref{dybe} implies that $T$ obeys the quadratic relation
	\[ \label{dyba}
	\begin{aligned}
	& R_{12} (\lambda_1 - \lambda_2, \theta-\gamma \,H ) \ T_1 (\lambda_1,\theta) \ T_2 (\lambda_2,\theta-\gamma\, h_1) \\
	& \qquad\qquad = \ T_2 (\lambda_2,\theta) \ T_1 (\lambda_1,\theta-\gamma \, h_2) \ R_{12} (\lambda_1 - \lambda_2,\theta) \ .
	\end{aligned}
	\]
It follows that the entries of $T$ as in \eqref{abcd} furnish $W$ with a representation of the operator algebra associated with $E_{\tau,\gamma}(\mathfrak{sl}_2)$, known as the \emph{(dynamical) Yang-Baxter algebra}, which we denote by $\mathfrak{A} = A_{\tau,\gamma}(\mathfrak{sl}_2)$. Since $T_0(\lambda,\theta) \in \End_\mathfrak{h} (V_0 \otimes W)$ preserves the grading by weights of $V_0 \otimes W$, and not that of $W$, the generators on the anti-diagonal of \eqref{abcd} are not $\mathfrak{h}$-invariant:
	\[ \label{abcd_ice}
	\begin{aligned}
	& [H,A(\lambda,\theta)] = 0 \ , && [H,B(\lambda,\theta)] = -2 \, B(\lambda,\theta) \ , \\
	& [H,C(\lambda,\theta)] = 2 \, C(\lambda,\theta) \ , \qquad && [H,D(\lambda,\theta)] = 0 \ .
	\end{aligned}
 	\]
Nevertheless the dynamical Yang-Baxter algebra is compatible with the grading of $W$ when $\mathfrak{A}$ itself is viewed as a graded algebra: $B$ and $C$ have weights ${-2}$ and $2$ respectively, while $A$ and $D$ have weight zero.

The generators contained in $\bar{T}$ are analogously seen to have weights $\mathrm{wt}(\bar{B})=-2$, $\mathrm{wt}(\bar{C})=2$ and $\mathrm{wt}(\bar{A})=\mathrm{wt}(\bar{D})=0$. The monodromy matrix $\bar{T}$ obeys
	\[ \label{dyba_bar}
	\begin{aligned}
	& R_{21} (\lambda_1 - \lambda_2, \theta) \ \bar{T}_1 (\lambda_1,\theta-\gamma \ h_2) \ \bar{T}_2 (\lambda_2,\theta) \\
	& \qquad\qquad = \ \bar{T}_2 (\lambda_2,\theta-\gamma\, h_1) \ \bar{T}_1 (\lambda_1,\theta) \ R_{21} (\lambda_1 - \lambda_2,\theta -\gamma \,H ) \ .
	\end{aligned}
	\]
This relation is equivalent to \eqref{dyba} due to the unitarity relations \eqref{unitarity} and \eqref{unitarity_mono}. In particular, the entries of $\bar{T}$ as in \eqref{abcd} subject to the relations encoded in \eqref{dyba_bar} also provide a representation of $\mathfrak{A}$. From \eqref{unitarity_mono} and \eqref{dyba} it follows that the generators $A,B,C,D$ and $\bar{A},\bar{B},\bar{C},\bar{D}$ are connected by \eqref{unitarity_mono} as well as by the relations contained in
	\[ \label{dyba_ttbar}
	\begin{aligned}
	& T_1 (\lambda_1,\theta-\gamma \,h_2) \ R_{12} (\lambda_1+\lambda_2,\theta) \ \bar{T}_2 (\lambda_2,\theta-\gamma \,h_1) \\
	& \qquad\qquad = \ \bar{T}_2 (\lambda_2,\theta) \ R_{12} (\lambda_1+\lambda_2,\theta-\gamma \,H) \ T_1 (\lambda_1,\theta) \ .
	\end{aligned}
	\]

\subsection{Reflecting end: dynamical reflection algebra}\label{sec:reflection}

One can regard the Yang-Baxter algebra, or its dynamical version, as governing the bulk of the lattice. The reflecting boundary is described by an operator $K_0(\lambda,\theta)\in\End(V_0)$
satisfying the \emph{reflection equation}, also known as \emph{boundary Yang-Baxter equation},
	\[ \label{dREQ}
	\begin{aligned}
	& R_{12} (\lambda_1 - \lambda_2,\theta) \ K_1 (\lambda_1,\theta) \ R_{21} (\lambda_1 + \lambda_2,\theta) \ K_2 (\lambda_2,\theta) \\
	& \qquad\qquad = \ K_2 (\lambda_2,\theta) \ R_{12} (\lambda_1 + \lambda_2,\theta) \ K_1 (\lambda_1,\theta) \ R_{21} (\lambda_1 - \lambda_2,\theta) \ .
	\end{aligned}
	\]
Although we consider a dynamical \textit{R}-matrix, notice that there are no shifts in the dynamical parameter: this is just the ordinary reflection equation. We are interested in the solution of \eqref{dREQ} that lies in $\End_\mathfrak{h} (V_0)$, which was obtained in \cite{BPO_96} and reads
	\[ \label{kmat}
	K_0(\lambda,\theta) = \begin{pmatrix}
	k_{+}(\lambda,\theta) & 0 \\
	0 & k_{-}(\lambda,\theta)  \end{pmatrix}_{\!0} \ ,
	\]
with nonzero entries
	\[ \label{reflection}
	k_+(\lambda,\theta) = f(\zeta+\lambda) \, \frac{f(\theta+\zeta-\lambda)}{f(\theta+\zeta+\lambda)} \ , \qquad k_-(\lambda,\theta) = f(\zeta-\lambda) 
	\]
depending on a fixed \emph{boundary parameter} $\zeta\in\mathbb{C}$. These quantities can be depicted as
\begin{equation}
K_0(\lambda,\theta) \ = \
\begin{tikzpicture}[baseline={([yshift=-.5*11pt*.8]current bounding box.center)}, font=\scriptsize,scale=0.8,thick]
\draw[->] (-35:1.2) -- (0,0) -- (35:1.2);

\fill[gray!30] (0,-1) rectangle (-.15,1) ; \draw (0,-1) -- (0,1);

\node at (-62.5:.8) {$\theta$};
\node at (-35:1.5) {${-\lambda}$};
\end{tikzpicture}
\ ,	\qquad	k_\pm(\lambda,\theta) \ = \
\begin{tikzpicture}[baseline={([yshift=-.5*11pt*.8]current bounding box.center)}, font=\scriptsize,scale=0.8,thick]
\draw[->] (-35:1.2) -- (0,0) -- (35:1.2);

\fill[gray!30] (0,-1) rectangle (-.15,1) ; \draw (0,-1) -- (0,1);

\draw[thin,densely dotted,->] (0,-.4) arc(-90:-35:.4);

\node at (-35:1.5) {$\pm$};		
\node at (-62.5:.6) {$\lambda$};
\node at (-62.5:1.1) {$\theta$};
\node at (62.5:.8) {$\theta$};
\node at (35:1.5) {$\pm$};
\end{tikzpicture}
\ .
\end{equation}
The dynamical reflection equation \eqref{dREQ} then becomes
\begin{equation}
	\begin{tikzpicture}[baseline={([yshift=-.5*11pt*.8]current bounding box.center)}, font=\scriptsize,scale=0.8,thick]
		\begin{scope}[shift={(0,-.5)}]
			\draw[->] (-50:1.25) -- (0,0) -- (50:2.75);
			\node at (-70:.8) {$\theta$};
			\node at (-50:1.7) {${-\lambda_2}$};
		\end{scope}
		\begin{scope}[shift={(0,.5)}]
			\draw[->] (-15:2) -- (0,0) -- (15:2);
			\node at (-15:2.5) {${-\lambda_1}$};
		\end{scope}
		\fill[gray!30] (0,-1.75) rectangle (-.15,1.75); \draw (0,-1.75) -- (0,1.75);
	\end{tikzpicture}
	\ = \ \
	\begin{tikzpicture}[baseline={([yshift=-.5*11pt*.8]current bounding box.center)}, font=\scriptsize,scale=0.8,thick]
		\begin{scope}[shift={(0,.5)}]
			\draw[->] (-50:2.75) -- (0,0) -- (50:1.25);
			\node at (-50:3.2) {${-\lambda_2}$};
		\end{scope}
		\begin{scope}[shift={(0,-.5)}]
			\draw[->] (-15:2) -- (0,0) -- (15:2);
			\node at (-52.5:.8) {$\theta$};
			\node at (-15:2.5) {${-\lambda_1}$};
		\end{scope}
		\fill[gray!30] (0,-1.75) rectangle (-.15,1.75); \draw (0,-1.75) -- (0,1.75);
	\end{tikzpicture}
	\ .
\end{equation}

\paragraph{Sklyanin's monodromy matrix.} The \emph{(dynamical) double-row monodromy matrix} $\mathcal{T}_0(\lambda,\theta)\in\End_\mathfrak{h} (V_0\otimes W)$
is defined as the composition
	\[ \label{full_mono}
	\mathcal{T}_0(\lambda,\theta) \coloneqq T_0(\lambda,\theta) \ K_0(\lambda,\theta) \ \bar{T}_0(\lambda,\theta) \ = \ \
	\begin{tikzpicture}[baseline={([yshift=-.5*11pt*.8+8pt]current bounding box.center)}, font=\scriptsize,scale=0.8,thick]
		\draw[rounded corners=5pt]    (0,0) -- (.7,-.35) -- (2.6,-.35) (3.4,-.35) -- (5,-.35) node[right] {${-\lambda}$} ;
		\draw[rounded corners=5pt,->] (0,0) -- (.7,+.35) -- (2.6,+.35) (3.4,+.35) -- (5,+.35);
		
		\foreach \y in {-1,...,1} {
			\draw (3+.2*\y,-.35) node{$\cdot\mathstrut$};
			\draw (3+.2*\y,+.35) node{$\cdot\mathstrut$};
		}
		
		\fill[gray!30] (0,-1.25) rectangle (-.15,1.25) ; \draw (0,-1.25) -- (0,1.25);
		
		\draw[->] (1,-1.25) node[below]{$\mu_L$} -- (1,1.25);
		\draw[->] (2,-1.25) node[below]{$\mu_{L-1}$} -- (2,1.25);
		\draw[->] (4,-1.25) node[below]{$\mu_1$} -- (4,1.25);
		
		\node at (.5,-.75) {$\theta$}; 
		\node at (1.5,-.75) {$\theta\mp\gamma$};
		\node at (2.5,-.75) {$\cdots$};
	\end{tikzpicture}
	\ .
	\]
As in \cite{Skl_88} one can use the relations~\eqref{dyba}--\eqref{dyba_ttbar} for $\mathfrak{A}$ and the reflection equation~\eqref{dREQ} to show that this operator satisfies the following quadratic relation:
	\[ \label{dREA}
	\begin{aligned}
	& R_{12} (\lambda_1 - \lambda_2,\theta-\gamma \, H) \ \mathcal{T}_1 (\lambda_1,\theta) \ R_{21} (\lambda_1 + \lambda_2,\theta-\gamma \, H) \ \mathcal{T}_2 (\lambda_2,\theta)  \\
	& \qquad\qquad = \ \mathcal{T}_2 (\lambda_2,\theta) \ R_{12} (\lambda_1 + \lambda_2,\theta-\gamma \, H) \ \mathcal{T}_1 (\lambda_1,\theta) \ R_{21} (\lambda_1 - \lambda_2,\theta-\gamma \, H) \ .
	\end{aligned}
	\]
Note that, unlike \eqref{dyba} and \eqref{dyba_bar}--\eqref{dyba_ttbar}, the relation \eqref{dREA} does not involve shifts in the dynamical parameters of any monodromy matrix. This permits us to think of $\theta$ as a constant and suppress the dependence on this parameter when writing \eqref{full_mono} as a matrix in auxiliary space with entries in $\End(W)$:
	\[ \label{ABCD}
	\mathcal{T}_0(\lambda,\theta) = \begin{pmatrix}
	\mathcal{A}(\lambda) & \mathcal{B}(\lambda) \\
	\mathcal{C}(\lambda) & \mathcal{D}(\lambda) \end{pmatrix}_{\!0} \ .
	\]

\paragraph{Dynamical reflection algebra.} The entries of \eqref{ABCD} together with $H$, subject to the relations contained in \eqref{dREA} and the ice rule for $\mathcal{T}$, endow $W$ with a representation of the \emph{(dynamical) reflection algebra}, which we denote by $\mathfrak{B}$. For our purposes it is convenient to follow \cite{Skl_88} and replace the generator $\mathcal{D}$ of $\mathfrak{B}$ by the combination
	\[ \label{D_tilde}
	\tilde{\mathcal{D}}(\lambda) \coloneqq \mathcal{D}(\lambda) - \frac{f(\gamma) \, f(\theta-\gamma(H-1)+2\lambda)}{f(2\lambda+\gamma) \, f(\theta-\gamma(H-1))} \, \mathcal{A}(\lambda) \ .
	\]

Just like for $\mathfrak{A}$, the ice rule for the double-row monodromy matrix endows $\mathfrak{B}$ with a grading by weights: $\mathrm{wt}(\mathcal{B})=-2$, $\mathrm{wt}(\mathcal{C})=2$ and $\mathrm{wt}(\mathcal{A})=\mathrm{wt}(\mathcal{D})=\mathrm{wt}(\tilde{\mathcal{D}})=0$. That is, \eqref{abcd_ice} also holds for $\mathcal{A},\mathcal{B},\mathcal{C},\tilde{\mathcal{D}}$ and it follows that
	\[ \label{ABCD_ice}
	\begin{aligned}
	& f(H) \, \mathcal{A}(\lambda) = \mathcal{A}(\lambda) \, f(H) \ , &&  \mspace{1mu} f(H) \, \mathcal{B}(\lambda)  \mspace{1mu} = \mathcal{B}(\lambda) \, f(H-2) \ , \\
	& \, f(H) \, \mathcal{C}(\lambda) = \, \mathcal{C}(\lambda) \, f(H+2) \mspace{1mu}  \ , \qquad && f(H) \, \tilde{\mathcal{D}}(\lambda) = \tilde{\mathcal{D}}(\lambda) \, f(H) \ .
	\end{aligned}
	\]
[Thus, instead of including $H$ as one of the generators of $\mathfrak{B}$, one may equivalently view $\mathfrak{B}$ as a graded algebra, with grading by weights.] Finally, with the help of \eqref{dyba_ttbar} and \eqref{ABCD_ice} one can check that \eqref{crossing_mono} implies that $\mathcal{B}$ enjoys the crossing symmetry
	\[ \label{crossing_B}
	\mathcal{B}(-\lambda-\gamma) =  { -\frac{ f(2\lambda + 2\gamma) \, f(\theta + \zeta + \lambda)}{f(2\lambda) \, f(\theta + \zeta - \lambda - \gamma)} } \, \mathcal{B}(\lambda) \ .
	\]

\subsection{Domain walls and the partition function}
\label{sec:PF}

The final ingredient for the partition function are the domain-wall boundary conditions on the three remaining boundaries, which need to be chosen such that $\mathcal{Z}$ is nontrivial. We consider the case first introduced for vertex models in \cite{Kor_82} and extended to \textsc{sos} models in \cite{KZ_00}. These boundary conditions are included in the \textsc{qism} with the help of the following vectors in $W$.

\paragraph{Pseudovacua.} The highest- and lowest-weight vectors $\ket{0}, \ket{\bar{0}} \in W$ are defined as
	\[ \label{zero}	
	\ket{0} \coloneqq 
	e_+ \otimes \cdots \otimes e_+ \ , \quad  H \ket{0} = +L\ket{0} \ , \qquad \ket{\bar{0}} \coloneqq 
	e_- \otimes \cdots \otimes e_- \ , \quad  H \ket{\bar{0}} = -L\ket{\bar{0}} \ . 
	\]
These are also eigenvectors for the entries on the diagonal of the double-row monodromy matrix~\eqref{ABCD}. We will need the following `vacuum' eigenvalues:
	\[ \label{Lambda_def}
	\mathcal{A}(\lambda) \, \ket{0} = \Lambda_{\mathcal{A}} (\lambda) \, \ket{0} \ , \qquad
	\tilde{\mathcal{D}}(\lambda) \, \ket{0} = \Lambda_{\tilde{\mathcal{D}}} (\lambda) \, \ket{0} \ , \qquad
	\bra{\bar{0}} \, \mathcal{A}(\lambda) = \bar{\Lambda}_{\mathcal{A}} (\lambda) \,  \bra{\bar{0}} \ .
	\]
In Appendix~\ref{sec:vevs} we show that these eigenvalues are given by
	\[ 
	\label{Lambda}
	\begin{aligned}
	\Lambda_{\mathcal{A}} (\lambda) = \ & [\zeta+\lambda] \frac{[\theta+\zeta-\lambda]}{[\theta+\zeta+\lambda]} \prod_{j=1}^L  [\lambda-\mu_j+\gamma \,, \lambda+\mu_j+\gamma] \ , \\
	\Lambda_{\tilde{\mathcal{D}}} (\lambda) = \ & [\zeta-\lambda-\gamma] \frac{[2\lambda \,, \theta+\zeta+\lambda+\gamma \,,  \theta-L\gamma]}{[2\lambda+\gamma \,, \theta+\zeta+\lambda \,, \theta-(L-1)\gamma]} \prod_{j=1}^L [\lambda-\mu_j \,, \lambda+\mu_j] \ , \\
	\bar{\Lambda}_{\mathcal{A}} (\lambda) = \ & [\zeta-\lambda] \, \frac{[\gamma \,,  \theta+(L-1)\gamma-2\lambda]}{[2\lambda+\gamma \,, \theta+(L-1)\gamma]} \prod_{j=1}^L [\lambda-\mu_j+\gamma \,, \lambda+\mu_j+\gamma] \\
	& + [\zeta+\lambda+\gamma] \frac{[2\lambda \,, \theta+\zeta-\lambda-\gamma \,,  \theta+L\gamma]}{[2\lambda+\gamma \,, \theta+\zeta+\lambda \,, \theta+(L-1)\gamma]} \prod_{j=1}^L [\lambda-\mu_j \,, \lambda+\mu_j] \ ,
	\end{aligned}
	\]
where we abbreviate $[\lambda] \coloneqq f(\lambda)$ and $[\lambda_1 \,, \lambda_2 \,, \cdots] \coloneqq [\lambda_1] \, [\lambda_2] \cdots$. Observe that
	\[ \label{crossing_Lambda}
	\Lambda_{\mathcal{A}}(-\lambda-\gamma) = \frac{[ 2\lambda+\gamma \,, \theta + \zeta + \lambda \,, \theta - (L-1)\gamma ]}{ [2\lambda \,, \theta + \zeta - \lambda - \gamma \,, \theta - L \gamma] } \, \Lambda_{\tilde{\mathcal{D}}}(\lambda) \ .
	\]

\paragraph{Partition function.} For models with toroidal boundary conditions one is usually able to express the partition function as a trace of a product of transfer matrices, so that the evaluation of the partition function becomes an eigenvalue problem. This is not the case when other types of boundary conditions are considered. Vertex models with domain-wall boundary conditions~\cite{Kor_82} still admit an algebraic formulation, yet evaluating the partition function is not a priori an eigenvalue problem. In the case of domain-wall boundaries, possibly with a reflecting end, the partition function can be written as the expectation value of a product of operators. For \textsc{sos} models the situation is analogous, and the partition function of the inhomogeneous \textsc{sos} model with domain walls and one reflecting end, as in Figure~\ref{fig:part_fn}, is given by
	\[ \label{PF}
	\mathcal{Z}(\Vector{\lambda}) = \bra{\bar{0}} \ordprod_{1 \leq j \leq L} \! \mathcal{B}(\lambda_j) \ \ket{0} \ .
	\]
It depends on all parameters of the model: $L$ spectral parameters  $\Vector{\lambda}=(\lambda_1, \To , \lambda_L)$, the dynamical parameter~$\theta$, $L$ inhomogeneities~$\mu_j$, the crossing parameter~$\gamma$, the boundary parameter~$\zeta$, and the elliptic nome~$\E^{\I\pi\tau}$. (Throughout this work vectors in $\mathbb{C}^L$ are printed in boldface.) In the graphical notation \eqref{PF} is represented by the diagram on the right in Figure~\ref{fig:part_fn}.

To conclude this section we collect the properties of $\mathcal{Z}$ that are needed to characterize the partition function in our approach. Tracing back the dependence on the spectral and dynamical parameters it is clear that $\mathcal{Z}$, being polynomial in the statistical weight~\eqref{weights} and boundary weights~\eqref{reflection}, is meromorphic. Moreover, up to an overall factor originating  the denominator in \eqref{reflection}, it is entire in the spectral parameters:
\begin{lemma}[Polynomial structure] \label{lem:pol_1}
The partition function~\eqref{PF} can be written as
	\[
	\label{Zbar}
	\mathcal{Z}(\Vector{\lambda}) = \frac{ \bar{\mathcal{Z}}(\Vector{\lambda}) }{ \prod_{j=1}^{L} [\theta + \zeta + \lambda_j]} \ ,
	\]
where the normalized partition function~$\bar{\mathcal{Z}}$, viewed as a function of any single $\lambda_j$, is a theta function of order~$2(L+1)$ and norm~$(L-1)\gamma$. 
\end{lemma}
\noindent The basics of higher-order theta functions, which are also known as elliptic polynomials, are summarized in Appendix~\ref{sec:higher_order}. For a proof of Lemma~\ref{lem:pol_1} we refer to \cite{Fil_11}, in which a factorizing Drinfel'd twist \cite{MS_96,*ABFPR_00} is used to determine the dependence of the partition function on the spectral parameters. In fact, in \cite{Fil_11} it is shown that $\bar{\mathcal{Z}}$ can be further written as $\prod_{j=1}^{L} [2\lambda_j]$ times a higher-order theta function of order~$2(L-1)$ and norm~$(L-1)\gamma$ in each variable, but \eqref{Zbar} suffices for our purposes.

\begin{lemma}[Symmetric function] \label{lem:symm_1}
The partition function~\eqref{PF} is symmetric in the $L$ spectral parameters contained in $\Vector{\lambda}$.
\end{lemma}
\begin{proof}
The reflection algebra \eqref{dREA} includes the following relation, obtained 
using \eqref{ABCD_ice}:
	\[ \label{rel_BB}
	\mathcal{B}(\lambda_1) \, \mathcal{B}(\lambda_2) = \mathcal{B}(\lambda_2) \, \mathcal{B}(\lambda_1) \ .
	\]
The symmetry of \eqref{PF} is thus obvious.
\end{proof}
\noindent Although we will not need it let us further mention that the partition function is also symmetric in the inhomogeneity parameters~$\mu_j$ \cite{Fil_11}.

\begin{remark}[Crossing symmetry] \label{rmk:crossing}
Let us finally record that the partition function~\eqref{PF} inherits crossing symmetry \eqref{crossing_B} of $\mathcal{B}$: for any $1\leq j\leq L$ it satisfies
	\[
	\label{PF_crossing}
	\mathcal{Z}(\Vector{\lambda})\Big|_{\lambda_j\mapsto-\lambda_j-\gamma} = { -\frac{ [2\lambda_j + 2\gamma \,, \theta + \zeta + \lambda_j ]}{ [2\lambda_j \,, \theta + \zeta - \lambda_j - \gamma]} } \, \mathcal{Z}(\Vector{\lambda}) \ .
	\]
This property is not a necessary ingredient for our approach, yet it will be useful to get a better understanding of certain issues that we will encounter along the way.
\end{remark}


\section{Functional equations}
\label{sec:results}

In this section we use the dynamical reflection algebra to obtain a functional equation that determines the partition function~\eqref{PF}. Our main results are as follows.

\begin{theorem}[Functional equation] \label{thm:fun}
The partition function \eqref{PF} of the \textsc{sos} model with one reflecting end and domain-wall boundaries obeys the linear functional equation
	\[ \label{FE}
	\sum_{\nu=0}^L M_\nu(\lambda_0;\Vector{\lambda}) \, \mathcal{Z}(\lambda_0, \To , \widehat{\lambda_\nu} , \To , \lambda_L ) = 0 \ ,
	\]
where the caret indicates that the $\nu$\!'th spectral parameter is omitted, and the coefficients~$M_\nu$ feature the vacuum eigenvalues from \eqref{Lambda}:
	\begin{align}
	& M_0(\lambda_0;\Vector{\lambda}) \coloneqq \bar{\Lambda}_{\mathcal{A}} (\lambda_0) - \Lambda_{\mathcal{A}} (\lambda_0) \prod_{j=1}^{L} \frac{[\lambda_j-\lambda_0+\gamma \,, \lambda_j+\lambda_0]}{[\lambda_j-\lambda_0 \,, \lambda_j+\lambda_0+\gamma]} \ , \nonumber \\
	& M_i(\lambda_0;\Vector{\lambda}) \coloneqq \frac{[2\lambda_i \,, \gamma \,, \theta+(L-1)\gamma+\lambda_i-\lambda_0]}{[2\lambda_i+\gamma \,, \lambda_i-\lambda_0 \,, \theta+(L-1)\gamma]} \, \Lambda_{\mathcal{A}} (\lambda_i) \prod_{\substack{j=1 \\j \neq i}}^{L} \frac{[\lambda_j-\lambda_i+\gamma \,, \lambda_j+\lambda_i]}{[\lambda_j-\lambda_i \,, \lambda_j+\lambda_i+\gamma]} \label{M0Mi} \\
	& \hphantom{M_i} + \frac{[\gamma \,, \theta+(L-2)\gamma-\lambda_i-\lambda_0 \,, \theta-(L-1)\gamma]}{[\lambda_i+\lambda_0+\gamma \,, \theta+(L-1)\gamma \,, \theta-L\gamma]} \, \Lambda_{\tilde{\mathcal{D}}} (\lambda_i) \prod_{\substack{j=1 \\j \neq i}}^{L} \frac{[\lambda_i-\lambda_j+\gamma \,, \lambda_i+\lambda_j+2\gamma]}{[\lambda_i-\lambda_j \,, \lambda_i+\lambda_j+\gamma]} \ . \nonumber
	\end{align}
\end{theorem}

In the terminology of \cite{Kuc_64}, \eqref{FE} is a cyclic linear
functional equation. Observe that it features $L+1$ variables whilst the partition function~\eqref{PF} depends on only $L$ spectral parameters. In analogy with the conventions from field theory in this work we use the Latin alphabet for indices ranging through $\{1,2,\To,L\}$ and Greek for indices in $\{0,1,2,\To,L\}$.

Theorem~\ref{thm:fun} asserts the \emph{existence} of solutions to the functional equation \eqref{FE}. In view of the linearity of \eqref{FE} the \emph{uniqueness} of its solutions has to be investigated as well. This issue is addressed by the following result.

\begin{theorem}[Uniqueness] \label{thm:unique}
Up to an overall $\Vector{\lambda}$-independent normalization factor, \eqref{FE} has a unique solution within the class of analytic functions on $\mathbb{C}^L$.
\end{theorem}

In fact, this uniqueness is established by exploiting the presence of the extra variable~$\lambda_0$ to find a recursion relation between the functional equations for length $L$ and $L-1$. This recursion moreover allows us to solve \eqref{FE} to obtain a multiple-integral formula for the partition function:

\begin{theorem}[Multiple-integral formula] \label{thm:sol}
The partition function~\eqref{PF} can be written as
	\[
	\label{sol}
	\begin{aligned}
	\mathcal{Z}(\Vector{\lambda}) = \ \Omega_L \, [\gamma]^L \, f'(0)^L \oint\limits_{\Gamma_{\Vector{\lambda}}^{\ \times L}} \! \frac{\D^L \Vector{z}}{(2\pi\I)^L} \, & \frac{ \prod_{\substack{i,j=1 \\ i\neq j}}^L [z_i-z_j] }{ \prod_{i,j=1}^L [z_i-\lambda_j]} \, \prod_{1\leq i<j\leq L} [z_i - \mu_j ,\, z_i + \mu_j + \gamma] \\
	& \qquad \times \prod_{l=1}^L  \frac{[2z_l]}{[2z_l+\gamma]} \, m_l(z_1,\To,z_l) \ ,
	\end{aligned}
	\]
where $\Omega_L$ is a constant (i.e.\ $\Vector{\lambda}$-independent) normalization factor, the integration contours $\Gamma_{\Vector{\lambda}}$ consist of small counter-clockwise loops enclosing only the poles of \eqref{sol} at $\lambda_j$, $1\leq j \leq L$, and the functions $m_l$ read
	\[
	\label{mi}
	\begin{aligned}
	m_l(z_1,\To,z_l) \coloneqq \ & [z_l+\zeta] \frac{ [\theta+\zeta-z_l \,, \theta+ l \gamma+z_l-\mu_l] }{ [\theta+\zeta+z_l \,, \theta+(l-1)\gamma] } [z_l+\mu_l+\gamma] \\
	& \qquad \times \prod_{j=1}^{l-1} [z_l-\mu_j+\gamma \,, z_l+\mu_j+\gamma] \frac{ [z_j-z_l+\gamma \,, z_j+z_l] }{ [z_j-z_l \,, z_j+z_l+\gamma] } \\ 
	& - [z_l-\zeta+\gamma] \frac{ [\theta+\zeta+z_l+\gamma \,, \theta+(l-1)\gamma-z_l-\mu_l] }{ [\theta+\zeta+z_l \,, \theta+(l-1)\gamma] } [z_l-\mu_l] \\
	& \qquad \times \prod_{j=1}^{l-1} [z_l-\mu_j \,, z_l+\mu_j] \frac{ [z_l-z_j+\gamma \,, z_l+z_j+2\gamma] }{ [z_l-z_j \,, z_l+z_j+\gamma] } \ .
	\end{aligned}
	\]
\end{theorem}

Direct inspection further suggests that the normalization constant in \eqref{sol} is given by
	\[
	\label{OmegaL}
	\Omega_L = \frac{[\theta+(L+1)\gamma]}{[\theta+ L \gamma]} \prod_{i=1}^L \frac{[\zeta-\mu_i \,, \theta+\zeta+\mu_i]}{[\zeta+\mu_i \,, \theta+\zeta-\mu_i]} \prod_{K = 0}^{\lfloor L/2 \rfloor} \frac{[\theta-(L-2 K)\gamma]}{[\theta+(L-2 K +1)\gamma]} \ ,
	\]
where $\lfloor L/2 \rfloor$ is the integer part of $L/2$.

Let us pause for a moment to look at the analytic structure of \eqref{sol}--\eqref{mi}. Although \eqref{sol} has apparent poles at $z_l = -\gamma/2$, these simple poles are removable as the residues of $m_l(z_1,\To,z_l)/[2z_l+\gamma]$ at these points vanish. The simple poles in \eqref{mi} at $z_l=z_j$ cancel against the corresponding zeroes in \eqref{sol}, while the poles at $z_l=-z_j-\gamma$ are related to those at $z_l=z_j$ by crossing symmetry, see Remark~\ref{rmk:crossing} above and Lemma~\ref{lem:M0Mi_symm} below. Finally the simple poles of \eqref{mi} at $z_l = -\theta-\zeta$ matches Lemma~\ref{lem:pol_1}.

\paragraph{Trigonometric degeneration.} Before turning to the proofs let us briefly discuss the limit $\tau\to\I\infty$ in which $f(\lambda) \propto \vartheta_1(\I\lambda|\tau)$ degenerates into a trigonometric function: $f(\lambda) \to \sinh(\lambda)$, cf.\ Appendix~\ref{sec:theta}. The trigonometric \textsc{sos} model with a diagonal reflecting end arises naturally by applying a vertex-face transformation to the six-vertex model with nondiagonal reflection~\cite{CLSW_03, *YZ_07}. The partition function of this model was studied in \cite{FK_10}, where it was expressed as a determinant. The set-up of Section~\ref{sec:setup} and all our results carry through upon reinterpreting $[\lambda] = \sinh(\lambda)$. In Lemma~\ref{lem:pol_1} the notion of elliptic polynomials has to be replaced by that of trigonometric (or Fourier) polynomials, which may equivalently be described as polynomials in $x=\E^{2\lambda}$ up to an overall factor of $\E^{-\lambda}$. The degree of the limiting trigonometric polynomial is at most equal to the order of an elliptic polynomial, though it can be less as the example $[2\lambda]$ illustrates.

When we further take the non-dynamical limit $\theta\to\infty$ our functional equation~\eqref{FE} and its solution~\eqref{sol}--\eqref{mi}, as well as the normalization~\eqref{OmegaL}, boil down to the results of \cite{GL_14}. Indeed, with our conventions all factors involving~$\theta$ come in ratios, and the expressions of \cite{GL_14} are heuristically recovered from those presented here by simply dropping all such ratios. [To match \eqref{mi} see Remark~\ref{rmk:alternative_2} at the end of Section~\ref{sec:sol} below.]

\subsection{Algebraic-functional method}
\label{sec:AF}

The algebraic-functional method \cite{Gal_08, Gal_10, Gal_11, Gal_12, Gal_13a, Gal_14, *Gal_13b, GL_14} exploits the (quantum-)algebraic structure to obtain functional equations. Here we apply this approach to derive the functional equation \eqref{FE} for the partition function~\eqref{PF} using the relations of the dynamical reflection algebra~$\mathfrak{B}$.

To keep track of the number of generators $\mathcal{A},\mathcal{B},\mathcal{C},\tilde{\mathcal{D}}$ that constitute any element of the dynamical reflection algebra we consider a grading $\mathfrak{B} = \bigoplus_{n\geq 0} \mathfrak{B}^{(n)}$ in addition to the grading by weights described in Section~\ref{sec:reflection}. This time we simply declare the generators on the right-hand side of \eqref{ABCD} to have degree one. Thus, the degree-$n$ part $\mathfrak{B}^{(n)} \subseteq \mathfrak{B}$ is the linear span of products of precisely $n$ generators $\mathcal{A},\mathcal{B},\mathcal{C},\tilde{\mathcal{D}}$. Note that this grading by number is compatible with the defining relations in \eqref{dREA} since those only involve terms of degree two.

\begin{proof}[Proof of Theorem~\ref{thm:fun}.]
The idea for obtaining our functional equation for the partition function \eqref{PF} is to use \eqref{Lambda_def} to insert an $\mathcal{A}$ on one side of the product of $\mathcal{B}$'s in \eqref{PF} and use the reflection-algebra relations to move it to the other side.

To save space let us indicate the arguments of the generators of $\mathfrak{B}$ as subscripts, e.g.\ $\mathcal{A}_0 \coloneqq \mathcal{A}(\lambda_0)$. We start by finding the appropriate algebraic relations in the dynamical reflection algebra. For this we need \eqref{rel_BB} together with the relations in $\mathfrak{B}^{(2)} \subseteq \mathfrak{B}$ obtained from
\eqref{dREA} using \eqref{ABCD_ice}:
	\[ \label{rel_AB}
	\begin{aligned}
	\mathcal{A}_0 \, \mathcal{B}_1 = \ & \frac{[\lambda_1-\lambda_0+\gamma \,, \lambda_1+\lambda_0]}{[\lambda_1-\lambda_0 \,, \lambda_1+\lambda_0+\gamma]} \, \mathcal{B}_1 \, \mathcal{A}_0 \\ 
	& - \frac{[\gamma \,, 2\lambda_1 \,, \theta-\gamma(H+1)+\lambda_1-\lambda_0]}{[\lambda_1-\lambda_0 \,, 2\lambda_1+\gamma \,, \theta-\gamma(H+1)]} \, \mathcal{B}_0 \, \mathcal{A}_1 \\
	& - \frac{[\gamma \,, \theta-\gamma(H+2)-\lambda_1-\lambda_0]}{[\lambda_1+\lambda_0+\gamma \,, \theta-\gamma(H+2)]} \, \mathcal{B}_0 \, \tilde{\mathcal{D}}_1 \ , 
	\end{aligned}
	\]
	\[ \label{rel_DB}
	\begin{aligned}
	\tilde{\mathcal{D}}_0 \, \mathcal{B}_1 = \ & \frac{[\lambda_0-\lambda_1+\gamma \,, \lambda_1+\lambda_0+2\gamma \,, \theta-\gamma H \,, \theta-\gamma(H+1)]}{[\lambda_0-\lambda_1 \,, \lambda_1+\lambda_0+\gamma \,, \theta-\gamma(H-1) \,, \theta-\gamma(H+2)]} \, \mathcal{B}_1 \, \tilde{\mathcal{D}}_0 \\
	& - \frac{[\gamma \,, 2\lambda_0+2\gamma \,, \theta-\gamma H \,, \theta-\gamma(H+1)+\lambda_0-\lambda_1]}{[\lambda_0-\lambda_1 \,, 2\lambda_0+\gamma \,, \theta-\gamma(H-1) \,, \theta-\gamma(H+2)]} \, \mathcal{B}_0 \, \tilde{\mathcal{D}}_1 \\
	& + \frac{[2\lambda_0+2\gamma \,, 2\lambda_1 \,, \gamma \,, \theta-\gamma H \,, \theta-\gamma H+\lambda_0+\lambda_1]}{[2\lambda_0+\gamma \,, 2\lambda_1+\gamma \,, \lambda_1+\lambda_0+\gamma \,, \theta-\gamma(H-1) \,, \theta-\gamma(H+1)]} \, \mathcal{B}_0 \, \mathcal{A}_1 \ ,
	\end{aligned} 
	\]
where $\tilde{\mathcal{D}}$ was defined in \eqref{D_tilde} such that no term proportional to $\mathcal{B}_1 \, \mathcal{A}_0$ appears on the right-hand side of \eqref{rel_DB}. By repeated application \eqref{dREA} gives rise to relations in $\mathfrak{B}^{(n)}$. In particular, \eqref{rel_BB} allows us to suppress the harpoon in \eqref{PF}, and from \eqref{rel_AB}--\eqref{rel_DB} it follows that in degree $n=L+1$ we have the following relation:
\begin{align}\label{rel_ABB}
	& \mathcal{A}_0 \, \prod_{j=1}^L \mathcal{B}_j = \prod_{j=1}^L \frac{[\lambda_j-\lambda_0+\gamma \,, \lambda_j+\lambda_0]}{[\lambda_j-\lambda_0 \,, \lambda_j+\lambda_0+\gamma]} \, \prod_{j=1}^L \mathcal{B}_j \ \mathcal{A}_0 \nonumber \\ 
	& \qquad\, - \sum_{i=1}^L \frac{[\gamma \,, 2\lambda_i \,, \theta-\gamma(H+1)+\lambda_i-\lambda_0]}{[\lambda_i-\lambda_0 \,, 2\lambda_i+\gamma \,, \theta-\gamma(H+1)]} \prod_{\substack{j=1\\j\neq i}}^L  \frac{[\lambda_j-\lambda_i+\gamma \,, \lambda_j+\lambda_i]}{[\lambda_j-\lambda_i \,, \lambda_j+\lambda_i+\gamma]} \, \prod_{\substack{\nu =0 \\ \nu\neq i}}^L \mathcal{B}_\nu \ \mathcal{A}_i \\
	& - \frac{[\theta-\gamma(H+2L-1)]}{[\theta-\gamma(H+2L)]} \sum_{i=1}^L \frac{[\gamma \,, \theta-\gamma(H+2)-\lambda_i-\lambda_0]}{[\lambda_i+\lambda_0+\gamma \,, \theta-\gamma(H+1)]} \prod_{\substack{j=1 \\ j\neq i}}^L \frac{[\lambda_i-\lambda_j+\gamma \,, \lambda_i+\lambda_j+2\gamma]}{[\lambda_i-\lambda_j \,, \lambda_i+\lambda_j+\gamma]} \, \prod_{\substack{\nu=0 \\ \nu\neq i}}^L \mathcal{B}_\nu \ \tilde{\mathcal{D}}_i \ . \nonumber
\end{align}

Now we multiply this relation from the left by $\bra{\bar{0}}$ and from the right by $\ket{0}$. This may be described more formally as follows. Consider the linear map $\mathrm{\pi} \colon \mathfrak{B} \longrightarrow \mathbb{C}[\lambda_0^{\pm1},\To,\lambda_L^{\pm1}]$ assigning multivariate meromorphic functions to elements of the reflection algebra defined by
	\[ \label{pi}
	\mathrm{\pi} (X) \coloneqq \bra{\bar{0}} X \ket{0} \ , \qquad\qquad X \in \mathfrak{B} \ .  
	\]
Write $\mathrm{\pi}_n \ \coloneqq \left.\mathrm{\pi}\right|_{\mathfrak{B}^{(n)}} \colon \mathfrak{B}^{(n)} \longrightarrow \mathbb{C}[\lambda_0^{\pm1},\To,\lambda_L^{\pm1}]$ for the restriction of $\mathrm{\pi}$ to the degree-$n$ part of $\mathfrak{B}$. Apply $\pi_{L+1}$ to \eqref{rel_ABB} and use \eqref{zero}--\eqref{Lambda_def} to express the result in terms of $\pi_{L}$. This yields the desired functional equation with coefficients \eqref{M0Mi}.
\end{proof}

A few comments are in order. Firstly, rather than inserting an $\mathcal{A}(\lambda_0)$ in \eqref{PF}, one could equally well use $\tilde{\mathcal{D}}(\lambda_0)$. In \cite{Gal_14, *Gal_13b} the corresponding functional equations are referred to as being of `type \textsc{a}' or `type \textsc{d}'. From the structural similarity between \eqref{rel_AB} and \eqref{rel_DB} it is clear that the equation of type \textsc{d} is also of the form~\eqref{FE} but with different coefficients. In the following section we will see that, like in \cite{GL_14}, the functional equation \eqref{FE}--\eqref{M0Mi} already suffices to characterize the partition function, and we do not need the explicit form of the functional equation of type~\textsc{d}.

Secondly, in view of the commutativity~\eqref{rel_BB} of the $\mathcal{B}$'s, the above proof also shows that the functional equation defined in \eqref{FE} is invariant under the interchange of variables $\lambda_i\leftrightarrow\lambda_j$ for $i,j\in\{1,\To,L\}$. Similarly, by \eqref{crossing_B} our equation is also invariant under crossing $\lambda_j\mapsto{-\lambda_j}-\gamma$ for all $j\in\{1,\To,L\}$. This agrees with the following observations. 

\begin{lemma}\label{lem:M0Mi_symm}
The coefficients \eqref{M0Mi}, viewed as functions $M_\nu \colon \mathbb{C}^{L+1} \to \mathbb{C}$, enjoy the following symmetry properties:
\begin{itemize}
	\item $M_0(\lambda_0;\Vector{\lambda})$ is symmetric in all $\lambda_j$ and invariant under $\lambda_j\mapsto{-\lambda_j}-\gamma$, for $1\leq j\leq L$;
	\item $M_i(\lambda_0;\Vector{\lambda})$ is symmetric in the $\lambda_j$ and invariant under $\lambda_j\mapsto{-\lambda_j}-\gamma$ for $j \in \{1,\To,L\}\setminus\{i\}$;
	\item $\left.M_i(\lambda_0;\Vector{\lambda})\right|_{\lambda_i\leftrightarrow\lambda_j} = M_j(\lambda_0;\Vector{\lambda})$, while
		\[
		\left.M_i(\lambda_0;\Vector{\lambda})\right|_{\lambda_i\mapsto{-\lambda_i}-\gamma} =  { -\frac{ [2\lambda_i + 2\gamma \,, \theta + \zeta + \lambda_i ]}{ [2\lambda_i \,, \theta + \zeta - \lambda_i - \gamma]} } \, M_i(\lambda_0;\Vector{\lambda}) \ .
		\]
\end{itemize}
\end{lemma}
\begin{proof}
The symmetry is clear from \eqref{M0Mi}; for crossing also use \eqref{crossing_Lambda}.
\end{proof}

Finally, when we instead exchange $\lambda_0\leftrightarrow\lambda_j$ for $j\in\{1,\To,L\}$ in \eqref{FE} we get another functional equation for $\mathcal{Z}$, which is of the same form \eqref{FE} but with different coefficients. For later use we record that, applying Lemma~\ref{lem:M0Mi_symm}, the coefficient of $\mathcal{Z}(\lambda_0,\To,\widehat{\lambda_\nu},\To,\lambda_L)$ in the equation obtained from \eqref{FE} by switching $\lambda_0\leftrightarrow\lambda_\rho$ for $0\leq \rho\leq L$ can be written in terms of the original coefficients \eqref{M0Mi} as
	\[
	\label{Mrhonu}
	M_{\rho,\nu}(\lambda_0,\lambda_1,\To,\lambda_L) = 
	\begin{cases}
		M_\nu(\lambda_\rho;\lambda_0,\To,\widehat{\lambda_\rho},\To,\lambda_L) & \qquad 0 \leq \rho < \nu \leq L \ , \\
		M_0(\lambda_\rho;\lambda_0,\To,\widehat{\lambda_\rho},\To,\lambda_L) & \qquad 0 \leq \rho = \nu \leq L \ , \\
		M_{\nu+1}(\lambda_\rho;\lambda_0,\To,\widehat{\lambda_\rho},\To,\lambda_L) & \qquad 0 \leq \nu < \rho \leq L \ .
	\end{cases}	
	\]
Thus, in fact, we obtain $L+1$ different functional equations, all satisfied by the partition function~\eqref{PF}. As the partition function certainly is nonzero the matrix with entries \eqref{Mrhonu} must have vanishing determinant. This can indeed be verified by direct inspection for given $L$.

\subsection{Reduction and uniqueness}
\label{sec:reduction}

In this and the next section we solve our functional equation. Our approach elaborates on the analysis developed in \cite{Gal_12, Gal_13a, GL_14}. As a byproduct we will be able prove Theorem~\ref{thm:unique}.

Let us reserve the symbol `$\mathcal{Z}$' for the partition function~\eqref{PF} and study the linear functional equation
	\[ \label{FE_v2}
	\sum_{\nu=0}^L M_\nu(\lambda_0;\Vector{\lambda}) \, Z(\lambda_0, \To , \widehat{\lambda_\nu} , \To , \lambda_L ) = 0
	\]
with coefficients~\eqref{M0Mi}. In view of Lemmas \ref{lem:pol_1} and~\ref{lem:symm_1} we are looking for solutions~$Z$ to this equation in the space of symmetric higher-order theta functions (up to an overall factor) on $\mathbb{C}^L$. We begin by collecting some useful properties of the functional equation \eqref{FE_v2}.

First, the properties from Lemmas \ref{lem:pol_1} and~\ref{lem:symm_1} (and Remark~\ref{rmk:crossing}) are in fact automatic when we look for sufficiently nice solutions:

\begin{lemma}[Properties of analytic solutions] \label{lem:sol_properties}
Any analytic solution $Z$ of \eqref{FE_v2} has the following properties:
\begin{enumerate}
	\item[i)] The normalized solution
		\[
		\label{Zbar_2}
		\bar{Z}(\Vector{\lambda}) \coloneqq Z(\Vector{\lambda}) \prod_{j=1}^{L} [\theta + \zeta + \lambda_j]
		\]
	is a theta function of order~$2(L+1)$ and norm~$(L-1)\gamma$ with respect to each $\lambda_j$. 
	\item[ii)] $Z$ is symmetric in the $L$ spectral parameters;
	\item[iii)] $Z$ satisfies the crossing symmetry
		\[
		\label{Z_crossing}
		Z(\Vector{\lambda})\Big|_{\lambda_j\mapsto-\lambda_j-\gamma} = -\frac{ [2\lambda_j + 2\gamma \,, \theta + \zeta + \lambda_j ]}{ [2\lambda_j \,, \theta + \zeta - \lambda_j - \gamma]} \, Z(\Vector{\lambda}) \ .
		\]
\end{enumerate}
\end{lemma}
\noindent The proofs can be found in Appendices \ref{sec:pol}--\ref{sec:crossing}.

Next, since \eqref{FE_v2} involves one more spectral parameter than $Z$ depends on, we can specialize any single spectral parameter to any value we like. Together, \eqref{Lambda} and \eqref{M0Mi} show that the greatest simplification of \eqref{FE_v2} occurs when we set $\lambda_0 = \pm\mu_k-\gamma$ for some $1\leq k \leq L$. Indeed, $M_0(\pm\mu_k-\gamma;\Vector{\lambda}) = \bar{\Lambda}_{\mathcal{A}}(\pm\mu_k-\gamma)$ consists of a single $\Vector{\lambda}$-independent product that is nonzero for generic values of the parameters.
This allows us to express $Z(\Vector{\lambda})$ as a linear combination of $Z$'s that each depend on only $L-1$ free spectral parameters. In particular, when $L=1$ this procedure completely determines the solution up to an overall scale:

\begin{lemma}[Case $L=1$] \label{lem:L=1}
For $L=1$ any solution to the functional equation \eqref{FE_v2} can be written as
 	\[ \label{sol_L=1}
 	Z(\lambda) = \Omega_1 \, [\gamma \,, \zeta+\mu_1] \frac{[\theta+\zeta-\mu_1 \,, \theta+\gamma]}{[\theta+\zeta+\lambda \,, \theta]} \, [2\lambda] = \Omega_1 \, [\gamma] \frac{[2\lambda]}{[2\lambda+\gamma]} \, m_1(z)
 	\]
for some normalization $\Omega_1$ not involving $\lambda$, where $m_1$ was defined in~\eqref{mi}. When $\Omega_1$ is given by \eqref{OmegaL} we recover the partition function~\eqref{PF}.
\end{lemma}
\begin{proof}
Set $\lambda_0= {\pm\mu_1} - \gamma$ in \eqref{FE_v2} to see that $Z(\lambda) \propto [2\lambda]/[\theta+\zeta+\lambda]$ for any proportionality constant not involving $\lambda$. The match with the partition function $\mathcal{Z}$ is found by direct inspection. [The second equality is recorded anticipating Section~\ref{sec:sol}.]
\end{proof}

For $L\geq 2$ we can further reduce the expression for $Z(\Vector{\lambda})$ by exploiting the following vanishing property for $Z$.

\begin{lemma}[Special zeroes] \label{lem:zeroes}
Suppose that $L\geq 2$ and fix $1\leq k\leq L$. Any solution of \eqref{FE_v2} vanishes at $\lambda_{L-1} = \lambda_+$ and $\lambda_L = \lambda_-$ for all four combinations of $\lambda_\pm \in \{ {\pm \mu_k - \gamma} , {\mp \mu_k} \}$.
\end{lemma}
\noindent Following \cite{Gal_12, Gal_13a, GL_14} we refer to these $4L$ combinations as `special zeroes' of the solution. Note that for given $k$ the four special zeroes are precisely related by crossing one or both values, $\mu\mapsto-\mu-\gamma$, cf.~Lemma~\ref{lem:sol_properties}~(iii). Lemma~\ref{lem:zeroes} is substantiated in Appendix~\ref{sec:zeroes}. Of course, by part~(ii) of Lemma~\ref{lem:sol_properties}, an analytic solution thus vanishes when \emph{any} two of its arguments are specialized to $\lambda_+$ and~$\lambda_-$. Together with parts~(i) and~(iii) of Lemma~\ref{lem:sol_properties} this has the following consequence.
\begin{corollary} \label{cor}
In the setting of Lemma~\ref{lem:zeroes} suppose that $Z$ is an analytic solution of \eqref{FE_v2}. Then we have for any $i\in\{1,\To,L\}$ that
	\[
	Z(\Vector{\lambda}) \Big|_{\lambda_i \in \{ {\pm \mu_k - \gamma} , {\mp \mu_k} \}} = \tilde{Z}_\pm(\lambda_1,\To,\widehat{\lambda_i},\To,\lambda_L) \, \prod_{\substack{j=1\\j\neq i}}^L [ \lambda_j \mp \mu_k \,, \lambda_j \pm \mu_k+\gamma ] \ , 
	\]
where the $\tilde{Z}_\pm$ are analytic functions in $L-1$ spectral parameters with the properties from Lemma~\ref{lem:sol_properties}, now with $L-1$ instead of $L$.
\end{corollary}

Although the preceding holds for any choice $\lambda_i \in \{ {\pm \mu_k - \gamma} , {\mp \mu_k} \}$ with any $1\leq k\leq L$, only when $k=1$ or $k=L$ it is possible to reduce the problem of solving \eqref{FE_v2} for length $L$ to that for $L-1$. The results for these two choices for $k$ turn out to be the same, cf.\ the proof of Theorem~\ref{thm:unique} and Remarks \ref{rmk:alternative_1} and~\ref{rmk:alternative_2} below, yet the case $k=L$ is easier to treat. The four resulting special zeroes can be treated simultaneously, yet for now we consider the two choices $(\mu_L-\gamma,\lambda_*)$ with $\lambda_* \in \{{-\mu_L}-\gamma,\mu_L\}$ to avoid cumbersome notation. This is consistent with crossing symmetry, implying that the result should not depend on this choice at any rate. We obtain the following recursive relation.

\begin{proposition}[Reduction] \label{prop:red}
If $Z$ is an analytic solution of the functional equation~\eqref{FE_v2} for length~$L$ then we can write 
	\[ \label{Z_via_Ztilde}
	Z (\Vector{\lambda}) = \, \Omega \, \sum_{i=1}^{L} M_i(\mu_L-\gamma;\Vector{\lambda}) \, \tilde{Z} (\lambda_1,\To,\widehat{\lambda_i},\To,\lambda_L) \, \prod_{\substack{j=1\\j\neq i}}^L [ \lambda_j-\mu_L \,, \lambda_j+\mu_L+\gamma ] \ ,
		\]
where $\tilde{Z}$ is an analytic solution of \eqref{FE_v2} for length~$L-1$, and the proportionality constant~$\Omega$ is independent of $\Vector{\lambda}$.
\end{proposition}

\begin{proof}
Let us employ the following shorthand for the arguments of various functions. We write e.g.\ $\lambda_{0,1,\To}$ for $\lambda_0,\lambda_1,\cdots$. The omission of a spectral parameter is indicated by a caret on the corresponding index; the presence of $\lambda_*$ is shown by an asterisk.

Suppose that $Z$ satisfies \eqref{FE_v2} for $L$. The idea is to use suitable specializations of $\lambda_0$ and $\lambda_L$. In view of the discussion preceding Lemma~\ref{lem:L=1} we first specialize $\lambda_0= \mu_L-\gamma$ in \eqref{FE_v2} and solve for $Z(\Vector{\lambda})$ to get
	\[
	Z(\Vector{\lambda}) = -\bar{\Lambda}_{\mathcal{A}}(\mu_L-\gamma)^{-1}  \sum_{i=1}^{L} M_i(\mu_L-\gamma;\Vector{\lambda}) \, Z(\lambda_{0,1,\To,\hat{\imath},\To,L})\big|_{\lambda_0= \mu_L -\gamma} \ .
	\]
Applying Corollary~\ref{cor} we obtain \eqref{Z_via_Ztilde} with $\Omega = -\bar{\Lambda}_{\mathcal{A}}(\mu_L-\gamma)^{-1}$. [An expression for $\tilde{Z}$ in terms of $Z$ can be found by setting $\lambda_L = \mu_L$ in \eqref{Z_via_Ztilde}.]

To see that $\tilde{Z}$ satisfies a functional equation of the form \eqref{FE_v2} for $L-1$ we proceed as follows. Carefully using \eqref{Z_via_Ztilde} to rewrite each $Z$ in our functional equation~\eqref{FE_v2} and specializing $\lambda_L = \lambda_* \in \{ -\mu_L -\gamma , \mu_L \}$ we find that $\tilde{Z}$ satisfies the functional equation
	\[ \label{FE_red}
	\sum_{\nu=0}^{L-1} \tilde{M}_\nu (\lambda_0;\lambda_{1,\To,L-1}) \, \tilde{Z}(\lambda_{0,1,\To,\hat{\nu},\To,L-1}) = 0 \ ,
	\]
with coefficients explicitly given in terms of \eqref{M0Mi} by
	\[ \label{M0Mi_red}
	\begin{aligned}
	\tilde{M}_\nu (\lambda_0;\lambda_{1,\To,L-1}) = \, & \Big( M_\nu(\lambda_0;\lambda_{1,\To,L-1,*}) \, 	M_L(\mu_L-\gamma;\lambda_{0,1,\To,\hat{\nu},\To,L-1,*}) \\
	& \ \, + M_L(\lambda_0;\lambda_{1,\To,L-1,*}) \, 	 M_{\nu+1}(\mu_L-\gamma;\lambda_{0,1,\To,L-1}) \Big) \\
	& \qquad \times \prod_{\substack{\rho= 0 \\ \rho\neq\nu}}^{L-1} [ \lambda_\rho-\mu_L \,, \lambda_\rho+\mu_L+\gamma ] \ .
	\end{aligned}
	\]
One can check that the two possible choices of $\lambda_*$ only lead to a constant ($\Vector{\lambda}$-independent) overall ($\nu$-independent) factor
for the coefficients \eqref{M0Mi_red}. Since this factor can be dropped in \eqref{FE_red} we have not indicated the dependence of the $\tilde{M}_\nu$ on this choice.

Clearly the reduced functional equation~\eqref{FE_red} has the same structure as our original functional equation when we take $L-1$ instead of $L$ in \eqref{FE_v2}; the coefficients \eqref{M0Mi_red} also exhibit the symmetries from Lemma~\ref{lem:M0Mi_symm} for length $L-1$. Although \eqref{FE_red} for $L=2$ happens to be proportional to \eqref{M0Mi} with $L=1$, unlike in earlier works such as \cite{Gal_12} (but like in \cite{GL_14}), this simple relation does unfortunately not persist for $L\geq 3$.

At the end of Section~\ref{sec:AF} we noticed that interchanging $\lambda_0 \leftrightarrow \lambda_j$ in \eqref{FE_v2} for some $1\leq j\leq L$ yields another functional equation. Thus, for length~\mbox{$L-1$}, in fact we have $L$ different functional equations: \eqref{FE_v2} together with the $L-1$ equations obtained by switching $\lambda_0 \leftrightarrow \lambda_j$. (Recall that not all of these equations are linearly independent.) We claim that \emph{for any length~$L$ the left-hand side of the reduced functional equation~\eqref{FE_red} can be written as some linear combination of the left-hand sides of these $L$ functional equations corresponding to length $L-1$.} To see that this is indeed the case we form the $L\times L$ matrix with entries given by \eqref{Mrhonu} for length~\mbox{$L-1$}, except that the last row is replaced by the coefficients \eqref{M0Mi_red} of the reduced equation. Our claim is true if this matrix has zero determinant. Whilst we have not managed to prove this rigorously, analytic and numerical investigations for $L\leq 7$ confirm that this is indeed the case for these system sizes, and we see no obstruction for this pattern to continue for larger $L$ too.
\end{proof}

Due to crossing symmetry, cf.~Lemma~\ref{lem:sol_properties}~(iii), the other special zeroes, $({-\mu_L}-\gamma,\lambda_*)$ with $\lambda_* \in \{\mu_L-\gamma,{-\mu_L}\}$, yield the same result. 
Another way to understand this is the fact that the proof of Proposition~\ref{prop:red} does not involve any algebraic choices, unlike e.g.\ for the equation $\lambda^2 = 1$, so that the procedure only has one possible result. (It may be instructive to convince oneself of this for $L=1$ and $L=2$.) Correspondingly it is now easy to prove the uniqueness, up to a constant overall factor, of analytic solutions of \eqref{FE_v2}.

\begin{proof}[Proof of Theorem~\ref{thm:unique}]
We use induction on $L$. The base case is furnished by Lemma~\ref{lem:L=1}. Suppose that $Z$ is an analytic solution of \eqref{FE_v2} for length $L$. Then by Proposition~\ref{prop:red} the function $\tilde{Z}$ in \eqref{Z_via_Ztilde} is analytic and solves the equation for length $L-1$. Hence, according to the induction hypothesis, $\tilde{Z}$ is unique up to a constant normalization factor. But \eqref{Z_via_Ztilde} determines $Z$ in terms of $\tilde{Z}$, again up to a constant factor. This proves the theorem.
\end{proof}

Note that the proof crucially depends on the recursion between the functional equation for successive lengths. In particular it also applies to the equations derived in \cite{Gal_12, Gal_13a, GL_14}.

\begin{remark}\label{rmk:alternative_1}
In view of the preceding proof one may wonder about the dependence on the choice $k=L$ in Proposition~\ref{prop:red}. The only other choice of $k$ that allows one to relate the reduced functional equation~\eqref{FE_red} to the equations for length one lower is $k=1$. In that case one finds that any analytic solution of \eqref{FE_v2} for length~$L$ can be written in terms of a solution $\check{Z}$ of \eqref{FE_v2} for length~$L-1$ and with inhomogeneities $(\mu_2, \To, \mu_L)$ instead of $(\mu_1, \To, \mu_{L-1})$:
	\[ \label{Z_via_Ztilde_alt}
	Z(\Vector{\lambda}) = \, \check{\Omega} \, \sum_{i=1}^{L} M_i(\mu_1-\gamma;\Vector{\lambda}) \, \Big( \check{Z}(\lambda_{1,\To,\hat{\imath},\To,L})\big|_{\mu_j\mapsto\mu_{j+1}}\Big) \, \prod_{\substack{j=1\\j\neq i}}^L [ \lambda_j-\mu_1 \,, \lambda_j+\mu_1+\gamma ] \ .
	\]
At the end of the next section we will see that the result of this alternative approach is the same as that for $k=L$, in accordance with Theorem~\ref{thm:unique}.
\end{remark}

\subsection{Multiple-integral formula}
\label{sec:sol}

In this section we solve our functional equation \eqref{FE_v2} by induction on $L$, based on the analysis of the previous section and Proposition~\ref{prop:red} in particular, culminating in the multiple-integral formula from Theorem~\ref{thm:sol}. 

As in the proof of Theorem~\ref{thm:unique} we proceed by induction on $L$. Iterating the relation from Proposition~\ref{prop:red} we can find a closed expression for the solution to \eqref{FE_v2}, which by Theorem~\ref{thm:unique} is unique up to an overall constant factor.
\begin{proposition}[Solution as symmetrized sum] \label{prop:sol}
Up to a constant overall normalization factor, the solution to the functional equation \eqref{FE_v2} for $L\geq 2$ can be written as the following symmetrized sum:
	\[ \label{sol_symm_sum} 
	\begin{aligned}
	Z(\lambda_1,\To,\lambda_L) = \ \sum_{\sigma \in S_L} & Z(\lambda_{\sigma(1)}) \prod_{l=2}^{L}  M_l(\mu_l-\gamma;\lambda_{\sigma(1)},\lambda_{\sigma(2)},\To,\lambda_{\sigma(l)}) \\
	& \qquad \times \prod_{1\leq i < j \leq L} [ \lambda_{\sigma(i)} - \mu_j \,, \lambda_{\sigma(i)} + \mu_j + \gamma ] \\
	= \ \sum_{\sigma \in S_L} & \left( \prod_{l=1}^{L} [\gamma] \frac{[2\lambda_{\sigma(l)}]}{[2\lambda_{\sigma(l)}+\gamma]} \,  m_l(\lambda_{\sigma(1)},\lambda_{\sigma(2)},\To,\lambda_{\sigma(l)}) \right) \\
	& \qquad \times \prod_{1\leq i < j \leq L} [ \lambda_{\sigma(i)} - \mu_j \,, \lambda_{\sigma(i)} + \mu_j + \gamma ] \ ,
	\end{aligned}
	\]
where $S_L$ denotes the symmetric group in $L$ symbols, $Z(\lambda)$ is given by \eqref{sol_L=1}, the factor $M_l$ is to be understood as given by \eqref{M0Mi} for length $l$, and $m_l$ was defined in \eqref{mi}.
\end{proposition}

\begin{proof}
The second equality follows directly from \eqref{Lambda}, \eqref{M0Mi} and \eqref{sol_L=1}. The proof of the first equality is by induction on $L$. For $L=2$ the statement follows directly from Proposition~\ref{prop:red}. The inductive step is straightforward, using Proposition~\ref{prop:red}, the properties of the coefficients $M_l$ recorded at the end of Section~\ref{sec:AF}, and the bijection of labelling sets $\{1,\To,L\} \times S_{L-1} \xrightarrow{\ \sim \ } S_L$ given by $(i,\sigma) \longmapsto (i,i+1,\To,L) \circ \sigma'$, where $\sigma'\in S_L$ is the extension of $\sigma$ fixing $L$. 
\end{proof}

The remaining step in the proof of Theorem~\ref{thm:sol} is based on the following trick, which appears to be common lore.

\begin{lemma}[Multiple-integral formula for symmetrized sums] \label{lem:sum_via_oint}
Consider $L\geq 1$ and let $\Vector{\lambda} \in \mathbb{C}^L$ be generic. Suppose that $g \colon \mathbb{C}^L \to \mathbb{C}$ is a meromorphic function that is regular in a neighbourhood of $z_i =\lambda_j$ for all $i,j\in\{1,\To,L\}$, and that $f\colon \mathbb{C}\to\mathbb{C}$ is analytic in a neighbourhood of the origin and satisfies $f(0)=0\neq f'(0)$. Then we can write
	\[
	\label{sum_via_oint} 
	\sum_{\sigma\in S_L} g(\lambda_{\sigma(1)},\To,\lambda_{\sigma(L)}) = f'(0)^L \oint\limits_{\Gamma_{\Vector{\lambda}}^{\, \times L}} \! \frac{\D^L \Vector{z}}{(2\pi\I)^L} \, \frac{\prod_{\substack{i,j=1 \\ i \neq j}}^L f(z_i-z_j)}{\prod_{i,j=1}^L f(z_i-\lambda_j)} \, g(\Vector{z}) \ ,
	\]
where each $z_i$ is integrated over the contour $\Gamma_{\Vector{\lambda}}$ consisting of small counter-clockwise oriented loops around all the $\lambda_j$, $1\leq j \leq L$.
\end{lemma}

\begin{proof} 
Again one proceeds by induction on $L$. For $L=1$ the statement follows immediately from Cauchy's residue theorem for the single pole at $z=\lambda$, with residue~$1/f'(0)$.

For $L\geq 2$ we assume that $\Vector{\lambda} \in \mathbb{C}^L$ is such that $f(\lambda_i-\lambda_j)\neq 0$ for all $i\neq j$; in particular this means that the components~$\lambda_j$ should all be distinct. The inductive step entails applying the residue theorem to integrate over $z_L$, then employing the induction hypothesis~\eqref{sum_via_oint} to the $L$ functions $g(\Vector{z})|_{z_L=\lambda_i} \prod_{j=1}^{L-1}f(\lambda_i - z_j)$, and finally using $\{1,\To,L\} \times S_{L-1} \xrightarrow{\ \sim \ } S_L$ as in the proof of Proposition~\ref{prop:sol} above.
\end{proof}

It is now easy to prove our final result.
\begin{proof}[Proof of Theorem~\ref{thm:sol}]
The function $f(\lambda) \propto \vartheta_1(\I \lambda|\tau)$ and the function $g(\Vector{\lambda})$ determined by the second line in \eqref{sol_symm_sum} satisfy the assumptions of Lemma~\ref{lem:sum_via_oint}, allowing us to pass to repeated contour integrals.
\end{proof}

\begin{remark}\label{rmk:alternative_2}
If one proceeds along the lines of Remark~\ref{rmk:alternative_1} rather than Proposition~\ref{prop:red} one finds an equivalent symmetrized sum:
	\[
	\label{sol_symm_sum_alt} 
	\begin{aligned}
	Z(\lambda_1,\To,\lambda_L) = \ \sum_{\sigma \in S_L} & \left( \prod_{l=1}^{L-1} \left. M_1(\mu_1-\gamma;\lambda_{\sigma(l)},\lambda_{\sigma(l+1)},\To,\lambda_{\sigma(L)})\right|_{\mu_k\mapsto\mu_{k+l-1}} \right) Z(\lambda_{\sigma(L)}) \\
	& \qquad \times \prod_{L\geq i > j \geq 1} [ \lambda_{\sigma(i)} + \mu_j + \gamma \,, \lambda_{\sigma(i)} - \mu_j ] \\
	= \ \sum_{\sigma \in S_L} & \left( \prod_{l=1}^{L} [\gamma] \frac{[2\lambda_{\sigma(l)}]}{[2\lambda_{\sigma(l)}+\gamma]} \,  m_{L-l+1}(\lambda_{\sigma(L)},\lambda_{\sigma(L-1)},\To,\lambda_{\sigma(l)}) \right) \\
	& \qquad \times \prod_{L \geq i > j \geq 1} [ \lambda_{\sigma(i)} + \mu_j + \gamma \,, \lambda_{\sigma(i)} - \mu_j ] \ ,
	\end{aligned}
	\]
Here the $l$th factor of $M_1$ is understood to be given by \eqref{M0Mi} for length \mbox{$L-l+1$} and with all inhomogeneity parameters shifted as indicated. One can verify that this expression is equal to \eqref{sol_symm_sum}. The resulting representation of the partition function as a multiple contour integral does look slightly different from \eqref{sol}, and is immediately seen to generalize the result of \cite{GL_14}.
\end{remark}

\section{Conclusion}
\label{sec:conclusion}

In this work we consider the elliptic \textsc{sos} model with domain-wall boundaries and one reflecting end. Using the dynamical reflection algebra we derive a cyclic linear functional equation for the model's partition function (Theorem~\ref{thm:fun}). By construction this functional equation has a solution. We also prove its uniqueness, up to normalization, within the class of analytic functions on $\mathbb{C}^L$ (Theorem~\ref{thm:unique}). This allows us to obtain a closed expression for the partition function as a symmetrized sum (Proposition~\ref{prop:sol}), which can be concisely rewritten as a multiple contour integral (Theorem~\ref{thm:sol}). By taking the trigonometric limit analogous results are obtained for the trigonometric \textsc{sos} model on the same lattice. Along the way special attention is paid to the properties of our equation. Our results generalize those of Galleas and the author \cite{GL_14}.

The equation that we obtain is by no means the only one that can be found in this way. Indeed, as we mention in Section~\ref{sec:AF} our functional equation may be said to be of `type~\textsc{a}'. One can similarly derive functional equations of `type~\textsc{c}', which we expect to be significantly more involved (cf.\ \cite{Gal_12}), or of `type~\textsc{d}'. In each case the form is as in \eqref{FE} but with coefficients different from \eqref{M0Mi}. Since as in \cite{GL_14}, but unlike in e.g.~\cite{Gal_14, *Gal_13b}, our functional equation alone already determines the partition function (Theorem~\ref{thm:unique}) we do not pursue these other possibilities.

A common feature of our functional equation and others obtained via the algebraic-functional method \cite{Gal_10, Gal_11, Gal_12, Gal_13a, Gal_14, *Gal_13b, GL_14} is its structure \eqref{FE}: it can be described as a cyclic linear 
functional equation \cite{Kuc_64}. Of course the coefficients, presently \eqref{M0Mi}, differ from case to case. Another notable difference with the functional equation from \cite{Gal_13a} is that, due to the reflecting boundary, our equation does not determine the behaviour of the partition function on the dynamical parameter~$\theta$. This can already be anticipated at the algebraic level by comparing the $\theta$-dependence of the monodromy matrices in the defining relations \eqref{dyba} of the dynamical Yang-Baxter algebra~$\mathfrak{A} = A_{\tau,\gamma}(\mathfrak{sl_2})$ and \eqref{dREA} of the dynamical reflection algebra~$\mathfrak{B} \subseteq \mathfrak{A}$.

There are several problems left that we do not address in this work. These include the homogeneous limit, in which our expressions for the partition function might be easier to handle than the determinant formulas obtained by Filali and Kitanine \cite{FK_10, Fil_11}, the thermodynamic limit, which has been studied for the ordinary six-vertex model with the same boundary conditions in \cite{KR_15}, and the relation to combinatorial problems at the ice-like point~$\gamma = \I \pi/3$ \cite{Kup_95, Ros_09}. For the trigonometric \textsc{sos} model it should also be quite straightforward to convert our functional equation into a family of partial differential equations as in \cite{GL_14}, which might give new insights into the symmetries of the partition function. Another possible direction is to investigate if this method can also be applied to models associated with Lie algebras $\mathfrak{sl}_N$ of higher rank, or even Lie superalgebras. It would be interesting to return to such issues in the future.

\subsection*{Acknowledgements} 

I am indebted to W.~Galleas for bringing the problem tackled in this paper to my attention, and am grateful to him as well as G.~Arutyunov and A.~Henriques for useful discussions and comments on the manuscript. I also thank \textsc{desy} for the kind hospitality during the course of this work. This work is supported by the Netherlands Organization for Scientific Research (\textsc{nwo}) under the \textsc{vici} grant 680-47-602 and by the \textsc{erc} Advanced Grant 246974, \textit{Supersymmetry: a window to non-perturbative physics}. I further acknowledge the \textsc{d-itp} consortium, an \textsc{nwo} program funded by the Dutch Ministry of Education, Culture and Science (\textsc{ocw}).

\begin{center}
	$\cdot \ \ \cdot \ \ \cdot \ \ - \ \ - \ \ - \ \ \cdot \ \ \cdot \ \ \cdot$
\end{center}

\appendix

\section{Theta functions}

In this appendix we collect the properties of Jacobi theta functions and their higher-order versions that are needed for this work.

\subsection{The odd Jacobi theta function}
\label{sec:theta}

The elliptic solution of the dynamical Yang-Baxter equation~\eqref{dybe} features the function
	\[ \label{ell2} 
	\begin{aligned}
	f(\lambda) \coloneqq \ -\frac{\I}{2} \, \E^{-\I \pi \tau/4} \, \vartheta_1(\I\lambda|\tau) & = \frac{1}{2} \sum_{n \in \mathbb{Z}} (-1)^{n} \, \E^{\I n(n+1) \pi \tau} \, \E^{-(2n+1)\lambda} \\
	& = \sum_{n \in \mathbb{N}_{0} } (-1)^{n} \, \E^{\I n(n+1) \pi \tau} \sinh\!\big((2n+1)\lambda\big) \ .
	\end{aligned}
	\] 
The series in \eqref{ell2} converge absolutely when the elliptic nome $\E^{\I \pi \tau}$ satisfies $|\E^{\I \pi \tau}|< 1$. When $\E^{\I \pi \tau}$ tends to zero the theta function degenerates into a trigonometric function, and in fact $\lim_{\tau \to \I\infty} f (\lambda)= \sinh(\lambda)$. 

In this work we often use the short-hand $[\lambda]\coloneqq f(\lambda)$ and $[\lambda_1 \,, \lambda_2 \,, \cdots ] \coloneqq [\lambda_1 ] \, [\lambda_2] \cdots$. The properties that we need to work with the special function $f$ are as follows:
\begin{enumerate}
\item[i)] \textit{Analytic structure:} $f$ is entire (holomorphic on all of $\mathbb{C}$) and only has simple zeroes;
\item[ii)] \textit{Double quasiperiodicity:} 
	\[
	[\lambda+\I\pi] = -[\lambda] \ , \qquad [\lambda+\I\pi\tau] = -\E^{-2\lambda} \E^{-\I\pi\tau}  [\lambda] \ ;
	\]
\item[iii)] \textit{Oddness:} 
	\[
	[-\lambda]=-[\lambda] \ ;
	\] 
\item[iv)] \textit{Addition rule:}
	\[ \label{addition_rule}
	\begin{aligned}
	& [\lambda_1+\lambda_3 \,, \lambda_1-\lambda_3 \,, \lambda_2+\lambda_4 \,, \lambda_2-\lambda_4] - [\lambda_1+\lambda_4 \,, \lambda_1-\lambda_4 \,, \lambda_2+\lambda_3 \,, \lambda_2-\lambda_3 ] \\
	& \qquad\qquad = [\lambda_1+\lambda_2 \,, \lambda_1-\lambda_2 \,, \lambda_3+\lambda_4 \,, \lambda_3-\lambda_4] \ .
	\end{aligned}
	\]
\end{enumerate}
Properties (i)--(iii) readily follow from the series \eqref{ell2}. The zeroes of $f$ form a lattice, $\I\pi\mathbb{Z} + \I\pi\tau\mathbb{Z} \subseteq \mathbb{C}$. Using Liouville's theorem in complex analysis it is easy to see that (i)--(iii) imply property (iv). Reversely, property (iv) implies (iii) and, up to a constant factor, (i) and (iv) uniquely characterize $f$ along with its trigonometric and rational limiting cases.

Since we occasionally encounter $f'(0)$ let us also mention the identities $\vartheta_1'(0|\tau) = \vartheta_2(0|\tau) \, \vartheta_3(0|\tau) \, \vartheta_4(0|\tau) = 2 \, \E^{\I \pi \tau/4} \, (\E^{2 \I \pi \tau}; \E^{2 \I \pi \tau})_\infty^{\, 3}$, with $(p;q)_\infty \coloneqq \prod_{n\in\mathbb{N}_0} (1-p\,q^n)$ the $q$-Pochhammer symbol. 
More about Jacobi theta functions can be found in \cite{MM_99, *WW_1927}.

\subsection{Higher-order theta functions}
\label{sec:higher_order}

Fix $\tau\in\mathbb{C}$ with $\Im(\tau)>0$ and consider the function $f$ from \eqref{ell2}. For $N\in\mathbb{N}_0$ and $t\in\mathbb{C}$ one defines a \emph{theta function of order~$N$ and norm~$t$} to be a complex function $F(\lambda)$ for which there exist numbers $\Omega,t_1,\To,t_N\in\mathbb{C}$ with $\sum_{n=1}^N t_n = t$ such that $F$ can be written in the factorized form 
	\[
	F(\lambda) = \Omega \prod_{n=1}^N [\lambda + t_n] = \Omega \, [\lambda + t_1 \,, \To , \lambda + t_N] \ .
	\]
Let $\Theta_{N\mspace{-1mu},\mspace{1mu}t}$ be the set of theta functions of order~$N$ and norm~$t$ with respect to the variable~$\lambda$. A classic result \cite[\textsection15]{Web_1891} is that $F(\lambda) \in \Theta_{N\mspace{-1mu},\mspace{1mu}t}$ if and only if $F(\lambda)$ is entire and doubly quasiperiodic with quasiperiods $\I\pi$ and $\I\pi\tau$ such that
	\[ \label{quasiper_orderN}
	F(\lambda + \I \pi) = (-1)^N F(\lambda) \ , \qquad F(\lambda + \I \pi \tau) = \E^{-2 t} \, (-\E^{-2\lambda} \E^{-\I\pi\tau})^N F(\lambda) \ .
	\]
For example one can verify that $F(\lambda) \coloneqq [n\lambda+\gamma]$ lies in $\Theta_{n^2\mspace{-1mu},\mspace{1mu}n\gamma}$ with respect to $\lambda$. As a corollary of~\eqref{quasiper_orderN} we see that $\Theta_{N\mspace{-1mu},\mspace{1mu}t}$ is a vector space: any linear combination of functions in $\Theta_{N\mspace{-1mu},\mspace{1mu}t}$ also satisfies \eqref{quasiper_orderN}. This factorization property for higher-order theta functions is very useful. 

For completeness we also mention that when $N\geq 2$ the dimension of $\Theta_{N\mspace{-1mu},\mspace{1mu}t}$ is equal to $N$, and that there is an interpolation formula expressing $F(\lambda) \in \Theta_{N\mspace{-1mu},\mspace{1mu}t}$ in terms of its values at $N$ generic points $\lambda_n \in \mathbb{C}$: 
	\[
	F(\lambda) = \sum_{n=1}^N  F(\lambda_n)  \frac{[\lambda - \lambda_n + t + \sum_{m=1}^N \lambda_m]}{[t + \sum_{m=1}^N \lambda_m]} \prod_{\substack{m=1\\m\neq n}}^N \frac{[\lambda-\lambda_m]}{[\lambda_n-\lambda_m]} \ .
	\]
Further details can be found in \cite{FS_99,PRS_08}.

\section{Vacuum eigenvalues}
\label{sec:vevs}

In this appendix we ascertain that the pseudovacua \eqref{zero} are eigenvectors of $\mathcal{A}(\lambda)$ and $\tilde{\mathcal{D}}(\lambda)$ and derive the expressions quoted in \eqref{Lambda} for the eigenvalues \eqref{Lambda_def}.

In terms of the ordinary monodromy matrices \eqref{abcd} and the reflection matrix \eqref{kmat} the entries \eqref{ABCD} of the double-row monodromy matrix \eqref{full_mono} read
	\[ \label{ABCD_via_abcd}
	\begin{aligned}
	\mathcal{A}(\lambda) & = k_{+}(\lambda,\theta) \, A(\lambda,\theta)\, \bar{A}(\lambda,\theta)  + k_{-}(\lambda,\theta) \, B(\lambda,\theta)\, \bar{C}(\lambda,\theta) \ , \\
	\mathcal{B}(\lambda) & = k_{+}(\lambda,\theta) \, A(\lambda,\theta)\, \bar{B}(\lambda,\theta)  + k_{-}(\lambda,\theta)  \,B(\lambda,\theta)\, \bar{D}(\lambda,\theta) \ , \\
	\mathcal{C}(\lambda) & = k_{+}(\lambda,\theta) \, C(\lambda,\theta)\, \bar{A}(\lambda,\theta)  + k_{-}(\lambda,\theta) \,D(\lambda,\theta)\, \bar{C}(\lambda,\theta) \ , \\
	\mathcal{D}(\lambda) & = k_{+}(\lambda,\theta) \, C(\lambda,\theta)\, \bar{B}(\lambda,\theta)  + k_{-}(\lambda,\theta) \,D(\lambda,\theta)\, \bar{D}(\lambda,\theta) \ .
	\end{aligned}
	\]
Thus we first compute the action of the generators of the dynamical Yang-Baxter algebra~$\mathfrak{A}$ on the vectors \eqref{zero}. Due to \eqref{rmat} and \eqref{weights}, $\ket{0}$ is a simultaneous eigenvector of $A,C,D$ and $\bar A,\bar C,\bar D$, with corresponding eigenvalues
	\[ \label{acd_action}
	\begin{aligned} 
	& \Lambda_A(\lambda,\theta) = \prod_{j=1}^L [\lambda-\mu_j+\gamma] \ , &&  \Lambda_{\bar{A}}(\lambda,\theta) = \prod_{j=1}^L [\lambda+\mu_j+\gamma] \ , \\ 
	& \Lambda_C(\lambda,\theta) = 0 \ ,  &&  \Lambda_{\bar{C}}(\lambda,\theta) = 0 \ ,  \\
	& \Lambda_D(\lambda,\theta) = \frac{[\theta+\gamma]}{[\theta-(L-1)\gamma]} \prod_{j=1}^L [\lambda-\mu_j] \ ,  && \Lambda_{\bar{D}}(\lambda,\theta) = \frac{[\theta-L\gamma]}{[\theta]} \prod_{j=1}^L [\lambda+\mu_j] \ .
	\end{aligned}
	\]
Likewise $\bra{\bar{0}}$ is an eigenvector of these operators, with eigenvalues
	\[ \label{acd_action2}
	\begin{aligned}
	& \bar\Lambda_A(\lambda,\theta) = \frac{[\theta-\gamma]}{[\theta+(L-1)\gamma]} \prod_{j=1}^L [\lambda-\mu_j] \ , && \bar\Lambda_{\bar{A}}(\lambda,\theta) = \frac{[\theta+ L\gamma]}{[\theta]} \prod_{j=1}^L [\lambda+\mu_j] \ , \\
	& \bar\Lambda_C(\lambda,\theta) = 0 \ ,  && \bar\Lambda_{\bar{C}}(\lambda,\theta) = 0 \ , \\
	& \bar\Lambda_D(\lambda,\theta) = \prod_{j=1}^L [\lambda-\mu_j+\gamma] \ , && \bar\Lambda_{\bar{D}}(\lambda,\theta) = \prod_{j=1}^L [\lambda+\mu_j+\gamma] \ . 
	\end{aligned}
	\]
Combining these with \eqref{ABCD_via_abcd} we directly find that $\Lambda_{\mathcal{A}}$ is given by the expression in \eqref{Lambda}. 
In contrast, neither pseudovacua is an eigenvector of $B$ or $\bar{B}$. This prevents a simple evaluation of $C\, \bar{B} \ket{0}$ needed for $\Lambda_{\mathcal{D}}$ in view of \eqref{ABCD_via_abcd}. This issue can be circumvented by using 
the relation \eqref{dyba_ttbar}, and \eqref{ABCD_ice}, to rewrite the problematic term as
	\[ \label{cbb}
	\begin{aligned}
	C(\lambda,\theta) \, \bar{B}(\lambda,\theta) = \ & \bar{B}(\lambda,\theta+\gamma) \, C(\lambda,\theta+\gamma) \\ & + \frac{[\gamma]}{[2\lambda+\gamma]} \, \bigg( \frac{[\theta-\gamma(H-1)+2\lambda]}{[\theta-\gamma(H-1)]} \, \bar{A}(\lambda,\theta+\gamma) \, A(\lambda,\theta+\gamma) \\ 
	& \hphantom{+ \frac{[\gamma]}{[2\lambda+\gamma]} \, \bigg( } \qquad - \frac{[\theta+\gamma+2\lambda]}{[\theta+\gamma]} \, D (\lambda,\theta) \, \bar{D}(\lambda,\theta) \bigg) \ .
	\end{aligned}
	\]
Together with \eqref{ABCD_via_abcd} and \eqref{acd_action} this yields the result for $\Lambda_{\tilde{\mathcal{D}}}$ from \eqref{Lambda}, where $\tilde{\mathcal{D}}$ was defined in \eqref{D_tilde} and we also used the addition rule \eqref{addition_rule} to rewrite the prefactor.

For $\bar\Lambda_{\mathcal{A}}$ we proceed analogously. The evaluation of $\bra{\bar{0}} B \, \bar{C}$ is avoided by exploiting the 
following relation contained in \eqref{dyba_ttbar}:
	\[ \label{bcc}
	\begin{aligned}
	B(\lambda,\theta) \, \bar{C}(\lambda,\theta) = \ & \bar{C}(\lambda,\theta-\gamma) \, B(\lambda,\theta-\gamma) \\
	& + \frac{[\gamma]}{[2\lambda+\gamma]} \bigg( \frac{[\theta-\gamma(H+1)-2\lambda]}{[\theta-\gamma(H+1)]} \, \bar{D}(\lambda,\theta-\gamma) \, D(\lambda,\theta-\gamma) \\ 
	& \hphantom{ + \frac{[\gamma]}{[2\lambda+\gamma]} \bigg( }\qquad - \frac{[\theta-\gamma-2\lambda]}{[\theta-\gamma]} \, A(\lambda,\theta) \, \bar{A}(\lambda,\theta) \bigg) \ .
	\end{aligned}
	\]
In combination with \eqref{ABCD_via_abcd} and \eqref{acd_action2} this establishes the last expression in \eqref{Lambda}, again invoking \eqref{addition_rule} for the prefactor.

%

\section{Properties of solutions}

In this appendix we derive various properties of solutions of our functional equation.

\subsection{Polynomial structure}
\label{sec:pol}

Here we investigate the analytic structure of any given solution of the functional equation~\eqref{FE_v2}. In view of Lemma~\ref{lem:pol_1} we expect a simple pole at ${-\theta}-\zeta$. Part~(i) of Lemma~\ref{lem:sol_properties} states that up to these poles, the solution is a higher-order theta function whose order and norm in the spectral parameters are determined by the functional equation.

\begin{proof}[Proof of Lemma~\ref{lem:sol_properties}~(i)]
Define $\bar{Z}$ as in \eqref{Zbar_2} and for each $\nu \in \{ 0,1,\To,L \}$ set
	\[
	\begin{aligned}
	\bar{M}_\nu(\lambda_0;\Vector{\lambda}) \coloneqq \ & [\theta+(L-1)\gamma \,, \theta+\zeta+\lambda_\nu] \,  M_\nu(\lambda_0;\Vector{\lambda})  \\ 
		& \qquad\qquad \times \prod_{\rho=0}^L [2\lambda_\rho+\gamma] \prod_{j = \rho+1}^L [\lambda_j - \lambda_\rho \,, \lambda_j + \lambda_\rho + \gamma]
	\end{aligned}
	\]
to remove all denominators. (These coefficients have the symmetry properties from Lemma~\ref{lem:M0Mi_symm}.) If $Z$ solves \eqref{FE_v2} then $\bar{Z}$ obeys the functional equation
	\[
	\label{FE_bar}
	\sum_{\nu = 0}^L \bar{M}_\nu(\lambda_0;\Vector{\lambda}) \, \bar{Z}(\lambda_0,\lambda_1,\To,\widehat{\lambda_\nu},\To,\lambda_L) = 0 \ . 
	\] 

First we focus on the dependence on $\lambda_0$. The explicit form of the coefficients, see \eqref{Lambda} and \eqref{M0Mi}, show that $\bar{M}_0$ has no poles in $\lambda_0$ whilst (for generic $\Vector{\lambda}$) the $\bar{M}_i$ only have simple zeroes in~$\lambda_0$. Therefore $\bar{Z}$ is entire in $\lambda_0$. Next one checks that $\bar{M}_0$ is a theta function of order~$4L+6$ and norm~$(L+2)\gamma-\theta$ whereas each $\bar{M}_i$ is a theta function of order~$2L+4$ and norm~$3\gamma-\theta$. Thus comparing the $\lambda_0$-dependence of the terms in \eqref{FE_bar} we conclude that $\bar{Z}$ must be a theta function of order~$2(L+1)$ and norm~$(L-1)\gamma$ in $\lambda_0$. 

One proceeds likewise for the dependence on $\lambda_i$ for any fixed $1\leq i\leq L$. With respect to this variable $\bar{M}_i$ is a theta function of order~$4L+6$ and norm~$(2L+1)\gamma$ while each $\bar{M}_\nu$ with $\nu\neq i$ is a theta function of order~$2L+4$ and norm~$(L+2)\gamma$. We conclude that $Z$ is also a theta function of order~$2(L+1)$ and norm~$(L-1)\gamma$ in the $\lambda_i$.
\end{proof}

It is worth pointing out that, unlike for the functional equation in \cite{Gal_12, Gal_13a}, similarly focussing on the dependence on the dynamical parameter~$\theta$ in \eqref{FE_bar} does not give us more information: the $\bar{M}_\nu$ are theta functions of the same order
and norm 
in~$\theta$ 
for each $\nu\in\{0,1,\To,L\}$.
This is in accordance with our notation, in which we do not indicate the dependence on $\theta$ of the coefficients or the solution (or partition function).

\subsection{Symmetric solutions}
\label{sec:symm}

In the proof of Theorem~\ref{thm:fun} we used the symmetry of the partition function $\mathcal{Z}$ from \eqref{PF}, cf.~Lemma~\ref{lem:symm_1}, to derive our functional equation. Here we show that although \eqref{FE_v2} is not manifestly symmetric in \emph{all} spectral parameters, cf.\ the end of Section~\ref{sec:AF}, reversely any analytic solution to \eqref{FE_v2} is symmetric.

\begin{proof}[Proof of Lemma~\ref{lem:sol_properties}~(ii)]
Consider the functional equation \eqref{FE_v2} in the limit $\lambda_0 \to \lambda_i$ for some $1 \leq i \leq L$. When $Z$ is analytic so that part~(i) of Lemma~\ref{lem:sol_properties} applies, \eqref{M0Mi} shows that the only singularities in \eqref{FE_v2} are the simple poles in $M_0$ and $M_i$, with opposite residues
	\[ \label{res}
	\begin{aligned}
	\Res_{\lambda_0 = \lambda_i\vphantom{^k}} M_0(\lambda_0;\Vector{\lambda}) & = {- \Res_{\lambda_0 = \lambda_i\vphantom{^k}} M_i(\lambda_0;\Vector{\lambda})} \\
	& = \frac{[2\lambda_i \,, \gamma]}{f'(0) \, [2\lambda_i+\gamma]} \Lambda_{\mathcal{A}} (\lambda_i) \prod_{\substack{j=1 \\ j\neq i}}^{L} \frac{[\lambda_j-\lambda_i+\gamma \,, \lambda_j+\lambda_i]}{[\lambda_j-\lambda_i \,, \lambda_j+\lambda_i+\gamma]} \ .
	\end{aligned}
	\]
Computing the residue of \eqref{FE_v2} as $\lambda_0 \to \lambda_i$ we thus obtain
	\[ \label{res1}
	\big( Z(\Vector{\lambda}) - Z(\Vector{\lambda}_{\sigma_i}) \big) \, \Res_{\lambda_0 = \lambda_i\vphantom{^k}} M_0(\lambda_0;\Vector{\lambda}) = 0 \ ,
	\]
where $\Vector{\lambda}_{\sigma_i} = (\lambda_i,\lambda_1,\To,\lambda_{i-1},\lambda_{i+1},\To,\lambda_L)$ is the image of $\Vector{\lambda}$ under the (right) action of the permutation $\sigma_i = (i,i-1,\To,2,1) \in S_L$: its $j$th component is $(\lambda_{\sigma_i})_j = \lambda_{\sigma_i(j)}$. Since \eqref{res} is generically nonzero and $Z$ is analytic, \eqref{res1} implies that $Z(\Vector{\lambda}) = Z(\Vector{\lambda}_{\sigma_i})$. Hence $Z$ is invariant under the cycle $\sigma_i$ of length~$i$. 

Now any two cycles of length $2$ and $L$ already generate the full permutation group~$S_L$. Using the above argument for $i=2$ and $i=L$ we may thus conclude that $Z$ is invariant under the action of $S_L$ permuting the variables $\lambda_j$, as we wanted to show.
\end{proof}

\subsection{Crossing symmetry}
\label{sec:crossing}

The final part of Lemma~\ref{lem:sol_properties} states that all solutions of the functional equation enjoy crossing symmetry. The proof is analogous to that of part~(ii) above.

\begin{proof}[Proof of Lemma~\ref{lem:sol_properties}~(iii)]
We investigate \eqref{FE_v2} as $\lambda_0 \to {-\lambda_i}-\gamma$ for some $1 \leq i \leq L$. If the solution $Z$ is analytic, by part~(i) of Lemma~\ref{lem:sol_properties} the sole singularities in \eqref{FE_v2} are again the simple poles in $M_0$ and $M_i$, with residues given by
	\[ \label{res_crossing}
	\begin{aligned}
	\Res_{\lambda_0 = {-\lambda_i}-\gamma\vphantom{^k}} M_0(\lambda_0;\Vector{\lambda}) & = { -\frac{ [2\lambda_i + 2\gamma \,, \theta + \zeta + \lambda_i] }{ [2\lambda_i \,, \theta + \zeta - \lambda_i - \gamma] } } \, \Res_{\lambda_0 = {-\lambda_i}-\gamma\vphantom{^k}} M_i(\lambda_0;\Vector{\lambda}) \ , \\
	\Res_{\lambda_0 = {-\lambda_i}-\gamma\vphantom{^k}} M_i(\lambda_0;\Vector{\lambda}) & = { -\frac{[ \gamma \,, \theta - (L-1)\gamma ]}{ f'(0) \,  [\theta - L \gamma] } } \, \Lambda_{\tilde{\mathcal{D}}}(\lambda_i) \prod_{\substack{j=1 \\ j\neq i}}^{L} \frac{[\lambda_i-\lambda_j+\gamma \,, \lambda_i+\lambda_j + 2\gamma]}{[\lambda_i-\lambda_j \,, \lambda_i+\lambda_j+\gamma]} \ .
	\end{aligned}
	\]
Here we used \eqref{crossing_Lambda}. Using the symmetry of $Z$ from part~(ii) of Lemma~\ref{lem:sol_properties} the residue of \eqref{FE_v2} as $\lambda_0 \to {-\lambda_i}-\gamma$ becomes
	\[ \label{res1_crossing}
	\left( { -\frac{ [2\lambda_i + 2\gamma \,, \theta + \zeta + \lambda_i] }{ [2\lambda_i \,, \theta + \zeta - \lambda_i - \gamma] } } \, Z(\Vector{\lambda}) - Z(\Vector{\lambda})\big|_{\lambda_i \mapsto {-\lambda_i}-\gamma} \right) \, \Res_{\lambda_0 = \lambda_i\vphantom{^k}} M_i(\lambda_0;\Vector{\lambda}) = 0 \ .
	\]
As \eqref{res_crossing} is generically nonzero and $Z$ is analytic, \eqref{res1_crossing} implies that $Z$ is invariant under crossing $\lambda_i \mapsto {-\lambda_i}-\gamma$, which is what we had to demonstrate.
\end{proof}

\subsection{Special zeroes}
\label{sec:zeroes}

Let $L\geq 2$. In this appendix we argue that for any $1\leq k\leq L$, a solution $Z$ of the functional equation~\eqref{FE_v2} vanishes when $\lambda_{L-1}=\lambda_+$ and $\lambda_L = \lambda_-$, for any choice of $\lambda_\pm \in \{ {\pm\mu_k} - \gamma , {\mp\mu_k}\}$. This is the content of Lemma~\ref{lem:zeroes}.

Let us use shorthand-notation analogous to that in the proof of Proposition~\ref{prop:red}. It is easy to see that $M_{L-1}(\lambda_0;\lambda_{1,\To,L-2,+,-}) = M_L(\lambda_0;\lambda_{1,\To,L-2,+,-}) = 0$ when $\lambda_\pm$ are as above. Setting $\lambda_{L-1}=\lambda_+$ and $\lambda_L = \lambda_-$ in \eqref{FE_v2} thus yields
	\[
	\label{FE**}
	\sum_{\nu=0}^{L-2} M_\nu(\lambda_0;\lambda_{1,\To,L-2,+,-}) \, Z(\lambda_{1,\To,\hat{\nu},\To,L-2,+,-}) = 0 \ .
	\]
Like we described at the end of Section~\ref{sec:AF} one can swap $\lambda_0\leftrightarrow\lambda_j$ in \eqref{FE**} to obtain $L-1$ further functional equations of a similar form. It follows that
	\[
	\begin{pmatrix}
	M_{0,0}(\lambda_{0,\To,L-2,+,-}) & \cdots & M_{0,L-2}(\lambda_{0,\To,L-2,+,-}) \\
	\vdots & \ddots & \vdots \\
	M_{L-2,0}(\lambda_{0,\To,L-2,+,-}) & \cdots & M_{L-2,L-2}(\lambda_{0,\To,L-2,+,-})
	\end{pmatrix}
	\begin{pmatrix}
	Z(\lambda_{1,\To,L-2,+,-}) \\ \vdots \\ Z(\lambda_{0,\To,L-3,+,-})
	\end{pmatrix} = 0 \ ,
	\]
where the matrix entries are given in \eqref{Mrhonu}.

Our task is to show that the determinant of this matrix is nonzero. Although we have not been able to prove this rigorously, analytic and numerical inspection for $L\leq 9$ reveal that this is indeed the case. We see no reason to doubt that this is the case for larger $L$ as well.


\addcontentsline{toc}{section}{References}
\bibliography{references2}

\end{document}